\documentclass[12pt]{article}
\usepackage{amsmath}
\usepackage{graphicx,psfrag,epsf}
\usepackage{enumerate}
\usepackage{natbib}
\usepackage{url} 

\usepackage{comment}
\usepackage{chngcntr}
\usepackage{amsmath,commath, amsfonts,amsthm,mathtools,mathrsfs,bm}
\usepackage{thmtools}
\usepackage{enumitem}
\usepackage{float}
\usepackage{booktabs}
\usepackage{multirow}
\usepackage[flushleft]{threeparttable}
\usepackage{xcolor}
\usepackage{hyperref}
\hypersetup{
	colorlinks = true,
	linkcolor = blue,
	citecolor = blue,
}

\newtheorem{theorem}{Theorem}
\newtheorem{lemma}[theorem]{Lemma}
\newtheorem{assumption}{Assumption}
\newtheorem{corollary}[theorem]{Corollary}
\newtheorem{proposition}[theorem]{Proposition}
\declaretheorem[style=definition]{example}

\renewcommand\thmcontinues[1]{Continued}

\numberwithin{theorem}{section}
\numberwithin{proposition}{section}

\newcommand{\blind}{0}

\begin{document}

\def\spacingset#1{\renewcommand{\baselinestretch}%
{#1}\small\normalsize} \spacingset{1}


\if0\blind
{
  \title{\bf Efficient and Robust Estimation of the Generalized LATE Model}
  \author{Haitian Xie\thanks{
    Email: hax082@ucsd.edu. The author is grateful to his advisors Graham Elliott and Yixiao Sun, who were gracious with their advice, support and feedback. The author also thanks Wei-Lin Chen, Yu-Chang Chen, and Kaspar W\"uthrich for helpful suggestions and constructive comments. This paper was previously circulated under the title ``Generalized Local IV with Unordered Multiple Treatment Levels: Identification, Efficient Estimation, and Testable Implication.'' All remaining errors are my own.}\hspace{.2cm}\\
    Department of Economics, University of California San Diego}
  \maketitle
} \fi

\if1\blind
{
  \bigskip
  \bigskip
  \bigskip
  \begin{center}
    {\LARGE\bf Title}
\end{center}
  \medskip
} \fi

\bigskip
\begin{abstract}
	This paper studies the estimation of causal parameters in the generalized local average treatment effect (GLATE) model, a generalization of the classical LATE model encompassing multi-valued treatment and instrument. We derive the efficient influence function (EIF) and the semiparametric efficiency bound (SPEB) for two types of parameters: local average structural function (LASF) and local average structural function for the treated (LASF-T). The moment condition generated by the EIF satisfies two robustness properties: double robustness and Neyman orthogonality. Based on the robust moment condition, we propose the double/debiased machine learning (DML) estimators for LASF and LASF-T. The DML estimator is semiparametric efficient and suitable for high dimensional settings. We also propose null-restricted inference methods that are robust against weak identification issues. As an empirical application, we study the effects across different sources of health insurance by applying the developed methods to the Oregon Health Insurance Experiment.

\end{abstract}

\noindent%
\vfill

\newpage
\spacingset{1.45} 
\section{Introduction}
\label{sec:intro}

Since the seminal works of \cite{imbens1994identification} and \cite{angrist1996identification}, the \emph{local average treatment effect} (LATE) model has become popular for causal inference in economics. Instead of imposing homogeneity of the treatment effects as in the classical instrumental variable (IV) regression model, the LATE framework allows the treatment effect to vary across individuals. Under the monotonicity condition, the average treatment effect can be identified for a subgroup of individuals whose treatment choice complies with the change in instrument levels.

The current form of the LATE model only accepts binary treatment variables. This restriction is inconvenient in many economic settings where the treatment is multi-leveled in nature. For example, parents select different preschool programs for their kids, schools assign students to different classroom sizes, families relocate to various neighborhoods in housing experiments, and people choose different sources of health insurance.
To apply the LATE model to these settings, researchers often need to redefine the treatment so that there are only two treatment levels. However, merging the treatment levels can complicate the task of program evaluation and dampen the causal interpretation of the estimates. As pointed out by \cite{kline2016evaluating}, if the original treatment levels are substitutes, then there is ambiguity regarding which causal parameters are of interest.
After merging the treatment levels, the heterogeneity in the treatment effect across different treatment levels is lost. 

This paper addresses the above issues by generalizing the LATE framework to incorporate the potential multiplicity in treatment levels directly. We call the new framework the \emph{generalized LATE} (GLATE) model. The main assumption of the GLATE model is the unordered monotonicity assumption proposed by \cite{heckman2018unordered}, which is a generalization of the monotonicity assumption in the binary LATE model.\footnote{To distinguish with the GLATE model, we sometimes use the terminology ``binary LATE model" to refer to the LATE model studied by \cite{imbens1994identification} and \cite{abadie2003semiparametric}.} 

We generalize the identification results in \cite{heckman2018unordered} to explicitly account for the presence of conditioning covariates, which is often important in practical settings. Recently, \cite{blandhol2022tsls} point out that linear TSLS, the common way to control for covariates in empirical studies, does not bear the LATE interpretation. The only specifications that have LATE interpretations are the ones that control for covariates nonparametrically. Therefore, it is essential from the causal analysis perspective to incorporate the covariates into the GLATE framework in a nonparametric way.

The causal parameters identifiable in the GLATE model include \emph{local average structural function} (LASF) and \emph{local average structural function for the treated} (LASF-T).
LASF is the mean \emph{potential outcome} for specific subpopulations. These subpopulations are defined by their treatment choice behaviors and are generalizations of the concepts \emph{always takers}, \emph{compliers}, and \emph{never takers} in the binary LATE model. The parameter LASF-T further restricts the subpopulation to exclude individuals who do not take up the treatment.

The paper is concerned with the econometric aspects of the GLATE model. The analysis begins by deriving \emph{efficient influence function} (EIF) and \emph{semiparametric efficiency bound} (SPEB) for the identified parameters. The calculation is based on the method outlined in Chapter 3 of \cite{bickel1993efficient} and \cite{newey1990semiparametric}. We then verify that the \emph{conditional expectation projection} (CEP) estimator \citep[e.g.,][]{chen2008semiparametric}, constructed directly from the identification result, achieves the SPEB and hence is semiparametric efficient. Using these results, we may efficiently estimate other important parameters of interest by the plug-in method since a standard delta-method argument preserves semiparametric efficiency.

The EIF not only facilitates the efficiency calculation but can also serve as the moment condition for estimation. This is because the EIF is mean zero by construction and is equal to the original identification result plus an adjustment term due to the presence of infinite-dimensional parameters. We show that the moment condition constructed from the EIF satisfies two related robustness properties: double robustness and Neyman orthogonality. Double robustness guarantees that the moment condition is correctly specified in a parametric setting even when some nuisance parameters are not. 

The Neyman orthogonality condition means that the moment condition is insensitive to the nuisance parameters. This condition is particularly useful when the conditioning covariates are of high dimension. To further utilize this condition, we study the \emph{double/debiased machine learning} (DML) estimator \citep{chernozhukov2018double} in the GLATE setting. Under certain conditions regarding the convergence rate of the first-step nonparametric estimators, the DML estimator is asymptotically normally uniformly over a large class of \emph{data generating processes} (DGPs). 

The weak identification issue is a practical concern of the GLATE model. This is because both the treatment and instrument are multi-valued, and hence the subpopulation on which LASF and LASF-T are defined can be small in size. To deal with this issue, we propose null-restricted test statistics in one-sided and two-sided testing problems. This procedure is the generalization of the well-known Anderson-Rubin (AR) test. We show that the proposed tests are consistent and uniformly control size across a large class of DGPs, in which the size of the subpopulation mentioned above can be arbitrarily close to zero.

The paper is organized as follows. The remaining part of this section discusses the literature. Section \ref{sec:setup} introduces the GLATE model and the nonparametric identification results. Section \ref{sec:spe1} calculates the EIF and SPEB. Section \ref{sec:robustness} discusses the robustness properties of the moment condition generated by the EIF. Section \ref{sec:weak} proposes inference procedures under weak identification issues. Section \ref{sec:empirical} presents the empirical application. Section \ref{sec:conclusion} concludes. The proofs for theoretical results in the main text are collected in Appendix \ref{app:proof}.

\subsection{Literature Review}

The GLATE model provides a way to conduct causal inference under endogeneity when the treatment is multi-valued and unordered. As mentioned above, the identification result (conditional on the covariates) is first established in \cite{heckman2018unordered} by using the unordered monotonicity condition. \cite{lee2018identifying} proposes another method of identification in a similar model of multi-valued treatment. Their method is concerned with continuous instruments, while the GLATE is framed in terms of discrete-valued instruments. When the treatment levels are ordered, \cite{angrist1995two} derives the identification and estimation results for the causal parameter, which is a weighted average of LATEs across different treatment levels.

The literature on semiparametric efficiency in program evaluation starts with the seminal work of \cite{hahn1998role}, which studies the benchmark case of estimating the \emph{average treatment effect} (ATE) under \emph{unconfoundedness}. For multi-level treatment, \cite{cattaneo2010efficient} studies the efficient estimation of causal parameters implicitly defined through over-identified non-smooth moment conditions. In the case where unconfoundedness fails and instruments are present, \cite{frolich2007nonparametric} calculates the SPEB for the LATE parameter, and \cite{hong2010semiparametric} extend to the estimation of parameters implicitly defined by moment restrictions. 
In a more general framework encompassing missing data, \cite{chen2004semiparametric} and \cite{chen2008semiparametric} studies semiparametric efficiency bounds and efficient estimation of parameters defined through overidentifying moment restrictions. However, there is currently no theoretical research on semiparametric efficient estimation in models that encompasses endogeneity and unordered multiple treatment levels. 

Several ways are available for calculating the EIF for semiparametric estimators, as illustrated by \cite{newey1990semiparametric} and \cite{ichimura2022influence}.
Semiparametric efficiency calculations can be used to construct robust (Neyman orthogonal) moment conditions. This method is illustrated in \cite{newey1994asymptotic} and \cite{chernozhukov2016locally}. Based on the Neyman orthogonality condition, \cite{chernozhukov2018double} introduces
 the DML method that suits high dimensional settings. This is because Donsker properties and stochastic equicontinuity conditions are no longer required in deriving the asymptotic distribution of the semiparametric estimator.

For testing the GLATE model, \cite{sun2020instrument} proposes a bootstrap test which is the generalization and improvement of the test studied by \cite{kitagawa2015test} in the binary LATE model.

The GLATE model has received attention in the recent empirical literature due to its ability to model multi-valued treatment. \cite{kline2016evaluating} evaluate the cost-effectiveness of Head Start, classifying Head Start and other preschool programs as different treatment levels against the control group of no preschool. \cite{galindo2020empirical} assesses the impact of different childcare choice in Colombia on children’s development. \cite{pinto2014noncompliance} studies the neighborhood effects and voucher effects in housing allocations using data from the Moving to Opportunity experiment. Our theoretical analysis of the GLATE model presents important tools for estimation and inference that can be applied to those empirical settings. 


\section{Identification in the GLATE Model} \label{sec:setup}

This section describes the \emph{generalized local average treatment effect} (GLATE) model, discusses identification of the \emph{local average structural function} (LASF) and other parameters, and introduces the notation.

\subsection{The model}

We assume a finite collection of instrument values $ \mathcal{Z} = \{z_1, \cdots, z_{N_Z}\}$ and a finite collection of treatment values $ \mathcal{T} = \{t_1, \cdots, t_{N_T}\}$, where $N_Z$ and $N_T$ are respectively the total number of instrument and treatment levels. The sets $\mathcal{T}$ and $\mathcal{Z}$ are categorical and unordered. The instrumental variable $Z$ denotes which of the $N_Z$ instrument levels is realized. The random variables $T_{z_1}, \cdots, T_{z_{N_Z}}$, each taking values in $\mathcal{T}$, denote the collection of potential treatments under each instrument status. Thus, the observed treatment level is the random variable $T = T_Z = \sum_{z \in \mathcal{Z}} \mathbf{1}\{Z=z\} T_z$. For each given treatment level $t \in \mathcal{T}$, there is a potential outcome $Y_t \in \mathcal{Y} \subset \mathbb{R}$. The observed outcome is denoted by $Y = Y_T = \sum_{t \in \mathcal{T}} \mathbf{1}\{T=t\} Y_t$. The random vector $X \in \mathcal{X} \subset \mathbb{R}^{d_X}$ contains the set of covariates. The observed data is a random sample $(Y_i,T_i,Z_i,X_i), 1 \leq i \leq n$. 

The description above establishes a random sampling model where the researcher only observes one potential outcome, the one associated with the observed treatment. This implies that the sample of $Y$, observed from an individual with treatment $T=t$, comes from the conditional distribution of $Y_t$ given $T=t$ rather than from the marginal distribution of $Y_t$. In general, this fact leads to identifications issues and presents challenges for causal inference. To overcome these problems, we impose further structures on the model. 

\begin{assumption}[Conditional Independence] \label{ass:CI} 
	$( \{Y_{t}: t \in \mathcal{T} \}, \{T_{z}: z \in \mathcal{Z} \} ) \perp Z \mid X$.
	
\end{assumption}

\begin{assumption}[Unordered Monotonicity] \label{ass:UM}
	For any $t \in \mathcal{T}, z,z' \in \mathcal{Z}$, either 
		\begin{align*}
			\mathbb{P}(\mathbf{1}\{T_{z} = t\} \geq \mathbf{1}\{T_{z'} = t\} \mid X) = 1 
		\end{align*}
		or
		\begin{align*}
			\mathbb{P}(\mathbf{1}\{T_{z} = t\} \leq \mathbf{1}\{T_{z'} = t\} \mid X) = 1.
		\end{align*}
\end{assumption}

Assumption \ref{ass:CI} and \ref{ass:UM} provide the multi-valued analog of Assumption 2.1 in \cite{abadie2003semiparametric}. Assumption \ref{ass:CI} restricts that the instrument $Z$ is independent with the potential treatments and outcomes once we condition on $X$. Assumption \ref{ass:UM} is the conditional version of the unordered monotonicity condition proposed by \cite{heckman2018unordered}. It means that when we focus on a particular treatment level $t$ and a pair $(z,z')$ of instrument values, the binary environment should satisfy the usual monotonicity constraint in the LATE model. Specifically, the unordered monotonicity condition requires that a shift in the instrument moves all agents uniformly toward or against each possible treatment value.\footnote{As pointed out by \cite{vytlacil2002independence}, the LATE monotonicity condition is a restriction across individuals on the relationship between different hypothetical treatment choices defined in terms of an instrument.} 

We define the type $S$ of an individual as the vector of the potential treatments, that is,
\begin{align*}
	S = (T_{z_1} \cdots, T_{z_{N_Z}})'.
\end{align*}
By construction, $S$ is not observed. Assumption \ref{ass:UM}, the unordered monotonicity condition, is essentially a restriction on $\mathcal{S} \equiv \textit{supp}(S)$, the support of $S$. Denote the elements in $\mathcal{S}$ by $s_1,\cdots,s_{N_S}$, where $N_S$ is the cardinality of $\mathcal{S}$. A convenient way to characterize $\mathcal{S}$ is by using the $N_Z \times N_S$ matrix $R \equiv (s_1,\cdots,s_{N_S})$. The matrix $R$ is referred to as the response matrix since it describes how each type of individuals' treatment choice responds to the instrument. 

The role of $S$ is to assist the identification of the counterfactual outcomes by dividing the population into a finite number of groups, where identification can be achieved within specific groups. Those groups are defined as follows. For $k = 0,\cdots,N_Z$, let $\Sigma_{t,k}$ be the set of types in which the treatment level $t$ appears exactly $k$ times. That is, 
\begin{align*}
	\Sigma_{t,k} \equiv \{ s \in \mathcal{S}  : \textstyle \sum_{i=1}^{N_Z} \mathbf{1}\{ s[i] = t \} = k \},
\end{align*}
where $s[i]$ denotes the $i$th element of the vector $s$.
In particular, the collection $\Sigma_{t,k}, k = 0,\cdots,N_Z$ forms a partition of $\mathcal{S}$.

For individuals with type $S$ in the same type set $\Sigma_{t,k}$, their treatment response in terms of $T=t$ is in a way homogeneous. Thus, it is easier intuitively to identify the marginal distribution of the potential outcome $Y_t$ within each $\Sigma_{t,k}$. More specifically, we define the \emph{local average structural functions} (LASF) and the \emph{local average structural functions for the treated} (LASF-T) as follows.
\begin{align*}
	\text{LASF: } \beta_{t,k} &\equiv \mathbb{E}[Y_t \mid S \in \Sigma_{t,k}], \\ 
	\text{LASF-T: } \gamma_{t,k} &\equiv \mathbb{E}[Y_t \mid S \in \Sigma_{t,k}, T=t].
\end{align*}
Before presenting the identification results for the above two classes of parameters, 
we illustrate the GLATE model in the following two examples.

\begin{example} [Binary LATE model] \label{eg:binary}
	In the binary LATE model of \cite{imbens1994identification}, there are two treatment levels $\mathcal{T} = \{0,1\}$ and two instrument levels $\mathcal{Z} = \{0,1\}$. There are three types: $\mathcal{S} = \{ s_1 = (0,0)', s_2 = (0,1)', s_3 = (1,1)' \}$, which are referred to in the literature as defiers, compliers, and always-takers, respectively. The type set $\Sigma_{1,0} = \{s_1\}$ contains the defiers, $\Sigma_{1,1} = \{s_2\}$ the compliers, and $\Sigma_{1,2} = \{s_3\}$ the always-takers. The response matrix is the following binary matrix
	\begin{align*}
		R = (s_1,s_2,s_3) = 
		\begin{pmatrix}
			0 & 0 & 1 \\
			0 & 1 & 1
		\end{pmatrix}.
	\end{align*}
	The local average treatment effect is the treatment effect for the compliers, which can be written as the difference between two LASFs:
	\begin{align*}
		\mathbb{E}[Y_1 - Y_0 \mid S = \text{compliers}] = \mathbb{E}[Y_1 - Y_0 \mid T_1 > T_0] = \beta_{1,1} - \beta_{0,1}.
	\end{align*}

\end{example}

\begin{example}[Three treatment levels and two instrument levels] \label{eg:3t2z}
	The simplest GLATE model (excluding the binary case in Example \ref{eg:binary}) has three treatment levels $\mathcal{T}=\{t_1,t_2,t_3\}$ and two instrument levels $\mathcal{Z}=\{z_1,z_2\}$. There are five types specified as the columns in the following response matrix 
	\begin{align*}
		R = (s_1,s_2,s_3,s_4,s_5) = 
		\begin{pmatrix}
			t_1 & t_2 & t_3 & t_1 & t_2 \\
			t_1 & t_2 & t_3 & t_3 & t_3
		\end{pmatrix}.
	\end{align*}
	In this example, a shift from $z_1$ to $z_2$ moves all agents uniformly toward the treatment level $t_3$. The type set $\Sigma_{t_1,2} = \{s_1\}$ contains the type that always choose the treatment $t_1$ and thus can be referred to as $t_1$-always taker. The same applies to $\Sigma_{t_2,2} = \{s_2\}$ and $\Sigma_{t_3,2} = \{s_3\}$. The type set $\Sigma_{t1,1} = \{s_4\}$ switches from $t_1$ to $t_3$ and hence can be considered as $t_1$-swticher (or $t_1$-compliter). Similarly, we can refer to $\Sigma_{t_2,1} = \{s_5\}$ as $t_2$-switcher and $\Sigma_{t_2,1} = \{s_5\}$ as $t_3$-switcher.
	This model is used in \cite{kline2016evaluating} to study the causal effect of the Head Start preschool program. The instrument indicates whether the household receives a Head Start offer, and the treatment levels are $t_1=$ Head Start, $t_2 = $ other preschool programs, and $t_3 = $ no preschool. The unordered monotonicity condition means that anyone who changes behavior as a result of the Head Start offer does so to attend Head Start.
\end{example}

\subsection{Identification Results} \label{sec:id}

We introduce some matrix notations related to the type $S$.  For each treatment level $t \in \mathcal{T}$, let $B_t$ be a binary matrix of the same dimension as the response matrix $R$ with each element of $B_t$ signifying whether the corresponding element in the response matrix is $t$. That is, $B_t[i,j]$, the $(i,j)$th element of $B_t$, is whether $T_{z_i}$ equals $t$ for the subpopulation $S=s_j$. Define $b_{t,k} \equiv \left(\mathbf{1}\{s_1 \in \Sigma_{t,k}\}, \cdots, \mathbf{1}\{s_{N_S} \in \Sigma_{t,k}\}  \right) B_t^+,$
where $B_t^+$ is the Moore-Penrose inverse of $B_t$.

For convenience, we also need some notations regarding conditional expectations. Let 
\begin{align*}
	\pi(X) \equiv ( \pi_{z_1}(X), \cdots, \pi_{z_{N_Z}}(X) )' \text{ with } \pi_z(X) \equiv P\left( Z =z \mid X \right)
\end{align*}
be the vector of functions that describes the conditional distribution of the instrument $Z$. For each treatment level $t \in \mathcal{T}$, let 
\begin{align*}
	P_{t}(X) \equiv (P_{t,z_1}(X) , \cdots, P_{t,z_{N_Z}}(X) )' \text{ with } P_t(X) \equiv P(T = t \mid Z=z,X)
\end{align*}
be the vector that describes the conditional treatment probabilities given each level of the instrument. Denote 
\begin{align*}
	Q_{t}(X) \equiv (Q_{t,z_1}(X) , \cdots, Q_{t,z_{N_Z}}(X) )' \text{ with } Q_{t,z}(X) \equiv \mathbb{E}[Y \mathbf{1}\{T=t\} \mid Z=z,X]
\end{align*}
as the vector that contains the conditional outcomes for each treatment level $t$. Notice that the functions $\pi$, $P_{t}$, and $Q_t$ are all identified.

\begin{theorem}[Identification of LASF]\label{thm:id1}
	Let Assumptions \ref{ass:CI} - \ref{ass:UM} hold. 
	Let $t \in \mathcal{T}$ and $k \in \{1,\cdots,N_Z\}$. 
	\begin{enumerate}[label = (\roman*)]
		\item The type set probability is identified by
		\begin{equation*}
			p_{t,k} \equiv \mathbb{P}(S \in \Sigma_{t,k}) =  b_{t,k} \mathbb{E} \left[ P_{t}(X)\right].
		\end{equation*}

		\item If $p_{t,k} > 0$, the LASF is identified by:
		\begin{equation*} \label{eqn:id_lasf}
			 \beta_{t,k} = b_{t,k} \mathbb{E} \left[  Q_{t}(X) \right] / p_{t,k}.
		\end{equation*}		
	\end{enumerate}
\end{theorem}

Theorem \ref{thm:id1} identifies $p_{t,k}$, the size of the subpopulation $\Sigma_{t,k}$, and the local structural function for that subpopulation. The only exception when the identification fails is when the type set $\Sigma_{t,0}$, in which case the individual never chooses the treatment $t$.
This identification result is a modification of Theorem T-6 in \cite{heckman2018unordered} that explicitly accounts for the presence of covariates $X$. Bayes rule is applied to convert the conditional result into the unconditional one. The following theorem presents the identification result for the LASF-T.



Let $\mathcal{Z}_{t,k} \subset \mathcal{Z}$ be the set of instrument values that induces the treatment level $t$ in the type set $\Sigma_{t,k}$. That is, $\mathcal{Z}_{t,k} \equiv \left\{ z_i \in \mathcal{Z} : s[i] = t \text{, for all } s \in \Sigma_{t,k} \right\}$, where $s[i]$ denotes the $i$th element of the vector $s$. Then define $\pi_{t,k} \equiv \sum_{z \in \mathcal{Z}_{t,k}} \pi_z$ as the total probability of those instrument values.

\begin{theorem}[Identification of LASF-T]\label{thm:id2}
	Let Assumptions \ref{ass:CI} - \ref{ass:UM} hold. 
	Let $t \in \mathcal{T}$ and $k \in \{1,\cdots,N_Z\}$. Then $\mathcal{Z}_{t,k}$ is nonempty.
	\begin{enumerate}[label = (\roman*)]
		\item The treatment probability within the type set is identified by
		\begin{equation*}
			q_{t,k} \equiv P\left( T = t, S \in \Sigma_{t,k}  \right) = b_{t,k} \mathbb{E} \left[   P_{t}(X) \pi_{t,k}(X)\right].
		\end{equation*}
		\item If $q_{t,k} >0 $, then the LASF-T is identified by
		\begin{equation} \label{eqn:id_lasft}
		\gamma_{t,k} = b_{t,k} \mathbb{E} \left[  Q_t(X) \pi_{t,k}(X) \right] / q_{t,k}.
	\end{equation}
\end{enumerate}
\end{theorem}

The identification results are illustrated using the two examples.

\begin{example}[continues = eg:binary] 
	Since the treatment is binary, the matrix $B_1$ is equal to the response matrix $R$. The matrix $B_1$ and its generalized inverse $B_1^+$ are respectively
	\begin{align*}
		B_1 = 
		\begin{pmatrix}
			0 & 0 & 1 \\
			0 & 1 & 1
		\end{pmatrix},
		\text{ and } 
		(B_1^+)' = 
		\begin{pmatrix}
			0 & -1 & 1 \\
			0 & 1 & 0
		\end{pmatrix}.
	\end{align*}
	The matrix $B_0$ and its generalized inverse $B_0^+$ are respectively
	\begin{align*}
		B_0 = 
		\begin{pmatrix}
			1 & 1 & 0 \\
			1 & 0 & 0
		\end{pmatrix},
		\text{ and } 
		(B_0^+)' = 
		\begin{pmatrix}
			0 & 1 & 0 \\
			1 & -1 & 0
		\end{pmatrix}.
	\end{align*}
	The vectors $b_{1,1}$ and $b_{0,1}$ are respectively 
	\begin{align*}
		b_{1,1} = (-1,1), \text{ and } b_{0,1} = (1,-1).
	\end{align*}
	Theorem \ref{thm:id1} implies that 
	\begin{align*}
		\beta_{1,1} = \frac{\mathbb{E}[Q_{1,1}(X)] - \mathbb{E}[Q_{1,0}(X)]}{\mathbb{E}[P_{1,1}(X)] - \mathbb{E}[P_{1,0}(X)]}, \text{ and } \beta_{0,1} = \frac{\mathbb{E}[Q_{0,0}(X)] - \mathbb{E}[Q_{0,1}(X)]}{\mathbb{E}[P_{0,0}(X)] - \mathbb{E}[P_{0,1}(X)]}.
	\end{align*}
	The two denominators in the above expressions are both equal to the type probability of compliers. Then the usual identification of the LATE parameter \citep[e.g.,][]{frolich2007nonparametric} follows:
	\begin{align*}
		\mathbb{E}[Y_1 - Y_0 \mid T_1 > T_0] = \frac{\int (\mathbb{E}[Y \mid Z=1,X=x] - \mathbb{E}[Y \mid Z=0,X=x]) f_X(x)dx}{\int (\mathbb{E}[T \mid Z=1,X=x] - \mathbb{E}[T \mid Z=0,X=x]) f_X(x)dx},
	\end{align*}
	where $f_X$ denotes the marginal density function of $X$.
\end{example}

\begin{example}[continues = eg:3t2z] 
	Recall that $\Sigma_{t_1,1} = \{s_4\}$ contains the $t_1$-switcher. By Theorem \ref{thm:id1}, the LASF for the treatment level $t$ and the subpopulation $S = s_4$ is identified by\footnote{The calculation of $b_{t,k}$ is omitted for brevity, but it can be done in the same way as Example \ref{eg:binary}.}
	\begin{align*}
		p_{t_1,1} &= \mathbb{E}[P_{t_1,z_1}(X)] - \mathbb{E}[P_{t_1,z_2}(X)], \\
		\beta_{t_1,1} &= \frac{\mathbb{E}[Q_{t_1,z_1}(X)] - \mathbb{E}[Q_{t_1,z_2}(X)]}{\mathbb{E}[P_{t_1,z_1}(X)] - \mathbb{E}[P_{t_1,z_2}(X)]}.
	\end{align*}
	Notice that $\mathcal{Z}_{t_1,1} = \{z_1\}$. Then by Theorem \ref{thm:id2} we have
	\begin{align*}
		q_{t_1,1} &= \mathbb{E}[(P_{t_1,z_1}(X) - P_{t_1,z_2}(X))\pi_{z_1}(X)], \\
		\gamma_{t_1,1} &= \frac{\mathbb{E}[(Q_{t_1,z_1}(X) - Q_{t_1,z_2}(X))\pi_{z_1}(X)]}{\mathbb{E}[(P_{t_1,z_1}(X) - P_{t_1,z_2}(X))\pi_{z_1}(X)]}.
	\end{align*}
\end{example}

\section{Semiparametric Efficiency} \label{sec:spe1}

In this section, we calculate the \emph{semiparametric efficiency bound} (SPEB) and propose estimators that achieve such bounds. We focus on the parameters LASF and LASF-T. In Appendix \ref{sec:spe2}, we study general parameters implicitly defined through moment restrictions.

\subsection{LASF and LASF-T}


For the rest of the paper, we assume that $Y_t, t \in \mathcal{T}$ have finite second moments. This is necessary since we are studying efficiency. 
Let $\iota$ denote the column vector of ones and $\zeta(Z,X,\pi)$ the diagonal matrix with the diagonal elements being $\mathbf{1}\{Z=z\}/\pi_{z}(X), z \in \mathcal{Z}$.
The following theorem gives the \emph{efficient influence function} (EIF) and the SPEB for the parameters identified in the preceding section. 

\begin{theorem} [SPEB for LASF and LASF-T] \label{thm:speb1}
	Let Assumptions \ref{ass:CI} - \ref{ass:UM} hold. 
	Let $t \in \mathcal{T}$ and $k \in \{1,\cdots,N_Z\}$.
	Assume that $p_{t,k}, q_{t,k} > 0$.
	\begin{enumerate}[label = (\roman*)]
		\item The semiparametric efficiency bound for $\beta_{t,k}$ is given by the variance of the efficient influence function
		\begin{align} \label{eqn:Psi_beta}
		\begin{split}
			&\psi^{\beta_{t,k}}(Y,T,Z,X,\beta_{t,k},p_{t,k},Q_t,P_{t},\pi)\\
			= & \frac{1}{p_{t,k}} b_{t,k} \left( \zeta(Z,X, \pi) \left( \iota (Y \mathbf{1}\{T=t\}) - Q_t(X) \right) + Q_t(X) \right) \\
			-& \frac{\beta_{t,k}}{p_{t,k}} b_{t,k} \left( \zeta(Z,X, \pi) \left( \iota \mathbf{1}\{T=t\} - P_t(X) \right) + P_t(X) \right).
		\end{split}
		\end{align}

		\item The semiparametric efficiency bound for $\gamma_{t,k}$ is given by the variance of the efficient influence function
		\begin{align*}
			\begin{split}
				&\psi^{\gamma_{t,k}}(Y,T,Z,X,\gamma_{t,k},q_{t,k},Q_t,P_t,\pi) \\
				=& \frac{1}{q_{t,k}} b_{t,k} \left( \zeta(Z,X, \pi) \left( \iota (Y \mathbf{1}\{T=t\}) - Q_t(X) \right)\pi_{{t,k}}(X) + Q_t(X)\mathbf{1}\{Z \in \mathcal{Z}_{t,k}\} \right) \\
				-& \frac{\gamma_{t,k}}{q_{t,k}} b_{t,k} \left( \zeta(Z,X, \pi) \left( \iota \mathbf{1}\{T=t\} - P_t(X) \right)\pi_{{t,k}}(X) + P_t(X)\mathbf{1}\{Z \in \mathcal{Z}_{t,k}\} \right).
			\end{split}
			\end{align*}

		\item The semiparametric efficiency bound for $p_{t,k}$ is given by the variance of the efficient influence function
		\begin{align*}
			\begin{split}
				\psi^{p_{t,k}}(T,Z,X,p_{t,k},P_t,\pi) = b_{t,k} \left( \zeta(Z,X, \pi) \left( \iota \mathbf{1}\{T=t\} - P_t(X) \right) + P_t(X)   \right)- p_{t,k}.
			\end{split}
		\end{align*}

		\item The semiparametric efficiency bound for $q_{t,k}$ is given by the variance of the efficient influence function
		\begin{align*}
			\begin{split}
				&\psi^{q_{t,k}}(T,Z,X,q_{t,k},P_t,\pi)\\
				=& b_{t,k} \left( \zeta(Z,X, \pi) \left( \iota \mathbf{1}\{T=t\} - P_t(X) \right)\pi_{{t,k}}(X) + P_t(X)\mathbf{1}\{Z \in \mathcal{Z}_{t,k}\} \right) - q_{t,k}.
			\end{split}
		\end{align*}
	\end{enumerate}
\end{theorem}

The EIF in Theorem \ref{thm:speb1} can be interpreted as the moment condition from the identification results modified by an adjustment term due to the presence of unknown infinite-dimensional parameters. Take $\psi^{\beta_{t,k}}$ as an example, the terms 
\begin{align*}
	b_{t,k} \left( \zeta(Z,X, \pi) \left( \iota (Y \mathbf{1}\{T=t\}) - Q_t(X) \right) \right) / p_{t,k}
\end{align*}
and 
\begin{align*}
	\beta_{t,k} b_{t,k} \left( \zeta(Z,X, \pi) \left( \iota \mathbf{1}\{T=t\} - P_t(X) \right) \right) / p_{t,k}
\end{align*}
are respectively the adjustment terms due to the presence of $Q_t$ and $P_t$. 

From the expression of $\psi^{\beta_{t,k}}$, we can see that the SPEB would be large when $p_{t,k}$ is small. This is because $p_{t,k}$ measures the size of the subpopulation $S \in \Sigma_{t,k}$ on which the LASF is estimated. When $p_{t,k}$ is small, we run into the weak identification issue. In Section \ref{sec:weak}, we study inference procedures that are robust against weak identification issues.

One benefit of the EIFs is that we can easily calculate the covariance matrix of different estimators. Consider an example where we are interested in two LASFs $\beta_1$ and $\beta_2$, whose EIF is given by $\psi_1$ and $\psi_2$, respectively. If the two estimators $\hat{\beta}_1$ and $\hat{\beta}_2$ are both semiparametric efficient, then their covariance matrix equals $\mathbb{E}[\psi_1\psi_2']$.

\begin{example}[continues = eg:binary]
	In the binary LATE model, the first two parts of Theorem \ref{thm:speb1} reduce to Theorem 2 of \cite{hong2010semiparametric}. If we assume unconfoundedness by having $T = Z$, then the result further reduces to Theorem 1 of \cite{hahn1998role}.
\end{example}

The derived SPEB helps determine whether an estimation procedure is efficient. In this section, we focus on the \emph{condition expectation projection} (CEP) estimator.\footnote{The terminology ``condition expectation projection'' is adopted from the papers \cite{chen2008semiparametric} and \cite{hong2010semiparametric}, whereas \cite{hahn1998role} refers to these estimators as ``nonparametric imputation based estimators.''} Define
\begin{align*}
	h_{Y,t,z}(X) = \mathbb{E} \left[ \mathbf{1}\{Z = z\}  Y \mathbf{1}\{T=t\} \mid X \right] \text{ and } h_{t,z}(X) = \mathbb{E} \left[ \mathbf{1}\{Z = z\}  \mathbf{1}\{T=t\} \mid X \right].
\end{align*}
The CEP procedure first estimates $\pi_z$, $h_{Y,t,z}$, and $h_{t,z}$ by using nonparametric estimators $\hat{\pi}_z$, $\hat{h}_{Y,t,z}$, and $\hat{h}_{t,z}$ respectively. These estimators can be constructed based on series or local polynomial estimation. Then $Q_{t,z}$ and $P_{t,z}$ are estimated using $\hat{Q}_{t,z} = \hat{h}_{Y,t,z} / \hat{\pi}_z$ and $\hat{P}_{t,z} = \hat{h}_{t,z} / \hat{\pi}_z$. The vectors of estimators $\hat{Q}_{t}$ and $\hat{P}_{t}$, $\hat{\pi}$ are stacked in an obvious way. Let $\hat{\pi}_{{t,k}} = \sum_{z \in \mathcal{Z}_{t,k}}\hat{\pi}_{z}$. The CEP estimators for the structural parameters are defined by
\begin{align*}
	\hat{p}_{t,k} & = \frac{1}{n}\sum_{i=1}^n b_{t,k} \hat{P}_{t}(X_i), && \hat{q}_{t,k} = \frac{1}{n}\sum_{i=1}^n b_{t,k} \hat{P}_{t}(X_i) \hat{\pi}_{{t,k}}(X_i),  \\
	\hat{\beta}_{t,k} & = \frac{1}{\hat{p}_{t,k}} \frac{1}{n} \sum_{i=1}^n b_{t,k} \hat{Q}_{t}(X_i), && \hat{\gamma}_{t,k} = \frac{1}{\hat{q}_{t,k}}\frac{1}{n}  \sum_{i=1}^n b_{t,k} \hat{Q}_{t}(X_i)  \hat{\pi}_{{t,k}}(X_i).
\end{align*}


The next proposition shows that the CEP estimators are semiparametrically efficient. The result is similar in style to \citeauthor{hahn1998role}'s (\citeyear{hahn1998role}) Proposition 4 that the low-level regularity conditions are omitted. Instead, the proposition assumes the high-level condition that the CEP estimators are asymptotically linear, which means they are asymptotically equivalent to sample averages. More formally, an estimator $\hat{\beta}$ of $\beta$ is asymptotically linear if it admits an influence function. That is, there exists an iid sequence $\psi_i$ with zero mean and finite variance such that
\begin{align*}
	\sqrt{n}(\hat{\beta} - \beta) = \frac{1}{\sqrt{n}}\sum_{i=1}^n \psi_i + o_p(1).
\end{align*}
Since each element of the conditional expectations $h_{Y,t,z}$, $h_{t,z}$, and $\pi_z$ can be considered as coming from a binary LATE model, the regularity conditions in \cite{hong2010supplement} should work with little modification.

\begin{proposition} \label{prop:cep}
	Suppose the CEP estimators are asymptotically linear, then they achieve the semiparametric efficiency bound.
\end{proposition}

The reason that this type of estimator is efficient is well explained in \cite{ackerberg2014asymptotic}. The estimation problem here falls into their general semiparametric model, where the finite-dimensional parameter of interest is defined by unconditional moment restrictions. They show that the semiparametric two-step optimally weighted GMM estimators, the CEP estimators in this case, achieve the efficiency bound since the parameters of interest are exactly identified. Discussions related to this phenomenon can also be found in \cite{chen2018overidentification}.

We next examine the efficient estimation of other policy-relevant parameters that can be derived from the parameters $\left( \beta_{t,k},\gamma_{t,k},p_{t,k},q_{t,k} \right) $. As an example, consider the type set $\Sigma_t \equiv \cup_{k=1}^{N_{Z-1}} \Sigma_{t,k}$, which is referred to as $t$-switchers. This subpopulation contains individuals who switch between $t$ and other treatments when given different levels of instruments. It is a generalization of the concept of compliers in the binary LATE framework.\footnote{Recall that switchers are also illustrated in Example \ref{eg:3t2z}.}  The LASF for the subpopulation $\Sigma_t$ is given by
\begin{align*}
	\beta_t \equiv \mathbb{E} \left[ Y_t \mid S \in \Sigma_t \right] = \frac{\sum_{k=1}^{N_Z-1} \beta_{t,k} p_{t,k}}{ \sum_{k=1}^{N_Z-1} p_{t,k}  }.
\end{align*}
Similarly, one can also define
\begin{equation}
	\gamma_t = \mathbb{E} \left[ Y_t \mid T=t, S \in \Sigma_t \right] = \frac{\sum_{k=1}^{N_Z-1} \gamma_{t,k} p_{t,k}}{ \sum_{k=1}^{N_Z-1} p_{t,k}  },
\end{equation}
which represents the LASF-T for the subpopulation of $t$-treated $t$-switchers.

For some subpopulations, a treatment effect can be identified. This point is already illustrated with Example \ref{eg:3t2z} in the discussion of the identification of the usual LATE parameter. We further illustrate this point with Example \ref{eg:3t2z}. 
\begin{example}[continues = eg:3t2z]
	The quantity
\begin{align*} \label{eqn:late}
\beta_{t_3,1} - \frac{\beta_{t_1,1} p_{t_1,1} + \beta_{t_2,1} p_{t_2,1}}{p_{t_1,1} + p_{t_2,1}} 
\end{align*}
represents the local average treatment effect of $t_3$ against other treatments within the subpopulation of $t_3$-switchers. Analogously, the parameter
\begin{align*}
\gamma_{t_3,1} - \frac{\gamma_{t_3,t_1,1} q_{t_3,t_1,1} + \gamma_{t_3,t_2,1} q_{t_3,t_2,1}}{q_{t_3,t_1,1} + q_{t_3,t_2,1}} 
\end{align*}
is the local average treatment effect of $t_3$ against other treatments within the subpopulation of $t_3$-treated $t_3$-switchers.
\end{example}

To summarize the above examples using a general expression, let $\phi = \phi(\underline{p},\underline{q},\underline{\beta},\underline{\gamma})$ be a finite-dimensional parameter, where $\phi(\cdot)$ is a known continuously differentiable function, and $\underline{p}$ is the vector containing all identifiable $p_{t,k}$'s, that is, $\underline{p} \equiv \{p_{t,k} : t \in \mathcal{T}, 1 \leq k \leq N_Z\}$. Let $\underline{q},\underline{\beta}$, and $\underline{\gamma}$ be defined analogously. A natural estimator can be defined through the CEP estimates, $\phi(\hat{\underline{p}},\hat{\underline{q}},\hat{\underline{\beta}},\hat{\underline{\gamma}})$. The delta method can help calculate the efficiency bound of $\phi$ and show the efficiency of $\phi(\hat{\underline{p}},\hat{\underline{q}},\hat{\underline{\beta}},\hat{\underline{\gamma}})$. In fact, by Theorem 25.47 of \cite{van1998asymptotic}, we immediately have the following corollary, which shows that plug-in estimators are efficient.

\begin{corollary} \label{co:1}
	The semiparametric efficiency bound of $\phi$ is given by the variance of efficient influence function
	\begin{equation}
		\psi^\phi = \sum_{p \in \underline{p}} \frac{\partial \phi}{\partial p} \psi^p + \sum_{q \in \underline{q}} \frac{\partial \phi}{\partial q} \psi^q + \sum_{\beta \in \underline{\beta}} \frac{\partial \phi}{\partial \beta} \psi^\beta + \sum_{\gamma \in \underline{\gamma}} \frac{\partial \phi}{\partial \gamma} \psi^\gamma
	\end{equation}
	where the partial derivatives are evaluated at the true parameter value. Moreover, the plug-in estimator $\phi(\hat{\underline{p}},\hat{\underline{q}},\hat{\underline{\beta}},\hat{\underline{\gamma}})$, based on the CEP estimators $\hat{\underline{p}},\hat{\underline{q}},\hat{\underline{\beta}},\hat{\underline{\gamma}}$, achieves the efficiency bound.
\end{corollary}

\section{Robustness} \label{sec:robustness}

In the previous section, the EIF is used as a tool for computing the SPEB. In this section, we directly use the EIF as the moment condition for estimation. These moment conditions are appealing because they satisfy double robustness and local robustness --- the two topics of this section. 

A word on notation: in the rest of the paper, we use a superscript $o$ to signify the true value whenever necessary. For example, when both $\pi^o$ and $\pi$ appear, the former means the true probability while the latter denotes a generic function.

\subsection{Double Robustness}

We focus on the LASF $\beta_{t,k}$. The same analysis can be applied to the other parameters. To avoid notational burden in the main text, we drop the subscript $(t,k)$ in $\beta_{t,k}$, $p_{t,k}$, and $b_{t,k}$, and the subscript $t$ in $P_t$ and $Q_t$.\footnote{The full subscripts are kept in the Appendices.}
It is straightforward to verify that the EIF $\psi^{\beta}$ has zero mean. However, we do not want to use $\psi^{\beta}$ itself as the estimating equation since it contains $1/p$ as a factor. To deal with this problem, we simply multiply $\psi^{\beta}$ by $p$ and define 
\begin{align*}
	\psi(Y,T,Z,X,\beta,Q,P,\pi) & = p \psi^{\beta}(Y,T,Z,X,\beta,p,Q,P,\pi) \\
	& = b \left( \zeta(Z,X, \pi) \left( \iota (Y \mathbf{1}\{T=t\}) - Q(X) \right) + Q(X) \right) \\
	& \quad -\beta b \left( \zeta(Z,X, \pi) \left( \iota \mathbf{1}\{T=t\} - P(X) \right) + P(X) \right).
\end{align*}
The corresponding moment condition is 
\begin{align} \label{eqn:robust-moment}
	\mathbb{E}\left[ \psi(Y,T,Z,X,\beta^o,Q^o,P^o,\pi^o) \right]=0 .
\end{align}
This moment condition is doubly robust, as demonstrated in the following proposition.


\begin{proposition} [Double Robustness] \label{prop:DR}
	Let $\left( Q,P,\pi \right)$ be an arbitrary vector of functions and $( Q^o,P^o,\pi^o ) $ the true vector of conditional expectations. Then 
	\begin{align*}
		\mathbb{E}\left[ \psi(Y,T,Z,X,\beta^o,Q^o,P^o,\pi) \right]=0 
	\end{align*}
	and
	\begin{align*}
		\mathbb{E}\left[ \psi(Y,T,Z,X,\beta^o,Q,P,\pi^o) \right]=0.
	\end{align*}
\end{proposition}

The above proposition divides the nonparametric nuisance parameters into two groups, $\pi$ and $(Q,P)$. The doubly robust moment condition is valid if either of these two groups of nuisance parameters is true. On the other hand, if the researcher uses parametric models for these nuisance parameters, then the structural parameter $\beta$ can be recovered provided that at least one of the working nuisance models is correctly specified. Therefore, the doubly robust moment condition is ``less demanding" on the researcher's ability to devise a correctly specified model for the nuisance parameters. The double robustness result in Proposition \ref{prop:DR} can be seen as the GLATE extension of the existing results in the binary LATE literature \citep[e.g.,][]{tan2006regression,okui2012doubly}.

\subsection{Neyman Orthogonality}

The second robustness property is Neyman orthogonality. Moment conditions with this property have reduced sensitivity with respect to the nuisance parameters. Formally, Neyman orthogonality means that the moment condition has zero Gateaux derivative with respect to the nuisance parameters. The result is presented in the following proposition.

\begin{proposition} [Neyman Orthogonality] \label{prop:Neyman}
	Let $\left( Q,P,\pi \right) $ be an arbitrary set of functions. For $r \in [0,1)$, define $Q^r = Q^o + r(Q - Q^o),$ $P^r = P^o + r(P - P^o),$ and $\pi^r = \pi^o + r(\pi - \pi^o)$. Suppose that $\sup_{r \in [0,1]}\big|\frac{\partial}{\partial r} \psi(Y,T,Z,X,\beta,Q^r,P^r,\pi^r) \big|$ is integrable, then
	\begin{align*}
		\frac{\partial}{\partial r} \mathbb{E} \left[ \psi(Y,T,Z,X,\beta,Q^r,P^r,\pi^r) \right] \Big|_{r=0} = 0,
	\end{align*}
	where $\beta$ does not need to be the true parameter value.
\end{proposition}

In many econometrics models, double robustness and Neyman orthogonality come in pairs. Discussions about their general relationships can be found in \cite{chernozhukov2016locally}. In practice, double robustness is often used for parametric estimation, as previously explained, whereas Neyman orthogonality is used in estimation with the presence of possibly high-dimensional nuisance parameters. 

Next, we apply the double/debiased machine learning (DML) method developed by \cite{chernozhukov2018double} to the moment condition (\ref{eqn:robust-moment}). This estimation method works even when the nuisance parameter space is complex enough that the traditional assumptions, e.g., Donsker properties, are no longer valid.\footnote{In two-step semiparametric estimations, Donsker properties are usually required so that a suitable stochastic equicontinuity condition is satisfied. See, for example, Assumption 2.5 in \cite{chen2003estimation}.} The implementation details are explained below.

The nuisance parameters $Q$, $P$, and $\pi$ are estimated using a cross-fitting method: Take an $L$-fold random partition of the data such that the size of each fold is $n/L$. For $l = 1, \cdots, L$, let $I_l$ denote the set of observation indices in the $l$th fold and $I^c_l = \bigcup_{l' \ne l} I_{l'}$ the set of observation indices not in the $l$th fold. Define $\check{Q}^l$, $\check{P}^l$, and $\check{\pi}^l$ to be the estimates constructed by using data from $I_l^c$. The DML estimator of $\beta$ is constructed following the moment condition (\ref{eqn:robust-moment}):\footnote{This is the DML2 estimator defined in \cite{chernozhukov2018double}. Another estimator, the DML1 estimator, is proposed in the same paper. We do not study the DML1 estimator since it is asymptotically equivalent to DML2, and the authors generally recommend DML2.}
\begin{align} \label{eqn:dml-beta}
	\check{\beta} = \frac{ \sum_{l=1}^L \sum_{i \in I_l}b \big( \zeta(Z_i,X_i, \check{\pi}^l) \big( \iota (Y_i \mathbf{1}\{T_i=t\}) - \check{Q}^l(X_i) \big) + \check{Q}^l(X_i) \big)}{ \sum_{l=1}^L \sum_{i \in I_l}b \big( \zeta(Z_i,X_i,\check{\pi}^l) \big( \iota \mathbf{1}\{T_i=t\} - \check{P}^l(X_i) \big) + \check{P}^l(X_i) \big)}.
\end{align}

To conduct inference, we also need an estimate for the asymptotic variance of $\check{\beta}$, which we denote by $\sigma^2$. The asymptotic variance equals to the expectation of the squared efficient influence function: $\sigma^2 = \mathbb{E}\left[ \psi^{\beta} \right]^2 = \mathbb{E}[\psi^2]/p^2$. We first estimate $p$ by using the cross-fitting method, which is essentially given by the denominator of (\ref{eqn:dml-beta}):
\begin{align} \label{eqn:dml-p}
	\check{p} = \frac{1}{n} \sum_{l=1}^L \sum_{i \in I_l} b \big( \zeta(Z_i,X_i, \check{\pi}^l) \big( \iota \mathbf{1}\{T_i=t\} - \check{P}^l(X_i) \big) + \check{P}^l(X_i)   \big).
\end{align}
Then the asymptotic variance can be estimated by
\begin{align*}
	\check{\sigma}^2 & = \frac{1}{n} \sum_{l=1}^L \sum_{i \in I_l} \big( \psi^{\beta}\big( Y_i,T_i,Z_i,X_i, \check{\beta}, \check{p}, \check{Q}^l, \check{P}^l, \check{\pi}^l \big) \big)^2 \\
	& = \frac{1}{n} \sum_{l=1}^L \sum_{i \in I_l} \big( \psi\big( Y_i,T_i,Z_i,X_i, \check{\beta}, \check{Q}^l, \check{P}^l, \check{\pi}^l \big) / \check{p} \big)^2.
\end{align*}

We want to establish the convergence results for the DML estimator uniformly over a class of data generating processes (DGPs) defined as follows.
For any two constants $c_1 > c_0 >0$, let $\mathcal{P}(c_1,c_0)$ be the set of joint distributions of $(Y,T,Z,X)$ such that 
\begin{enumerate} [label = (\roman*)]
	\item $p \in [c_0,1]$,
	\item $\mathbb{E}[\psi^2],\pi_z^o(X) \geq c_0,z \in \mathcal{Z}$, and $|Y \mathbf{1}\{T=t\}|,|Y \mathbf{1}\{T=t\} - Q_t^o(X)| \leq c_1$.
\end{enumerate}
The first condition excludes the case where $\beta$ is weakly identified (when $p$ can be arbitrarily close to zero). Inference under weak identification is studied in the next section.
The following theorem establishes the asymptotic properties of the DML estimation procedure. In particular, the estimator achieves the SPEB.

\begin{theorem} \label{thm:dml}
	Let Assumptions \ref{ass:CI} and \ref{ass:UM} hold.
	Assume the following conditions on the nuisance parameter estimators $(\check{Q}^l,\check{P}^l,\check{\pi}^l)$:
	\begin{enumerate}[label = (\roman*)]
		\item For $z \in \mathcal{Z}$, $|\check{Q}^l|$ is bounded, $\check{P}^l_{z}$ and $\check{\pi}^l_z \in [0,1]$, and $\check{\pi}^l_z$ is bounded away from zero.
		\item $\max_{z \in \mathcal{Z}} \big( \Vert \hat{Q} - Q^o \rVert_2 \vee \Vert \hat{P} - P^o \rVert_2 \vee  \lVert \hat{\pi} - \pi^o \rVert_2 \big) = o_p\big( n^{-1/4} \big)$.
	\end{enumerate}
	Then the estimator $\check{\beta}$ obeys that
	\begin{equation*}
		\sigma^{-1} \sqrt{n} \big( \check{\beta} - \beta \big) \Rightarrow N(0,1),
 	\end{equation*}
	uniformly over the DGPs in $\mathcal{P}(c_0,c_1)$. Moreover, the above convergence result continues to hold when $\sigma$ is replaced by the estimator $\check{\sigma}$. 
\end{theorem}

The proof verifies the conditions of Theorem 3.1 in \cite{chernozhukov2018double}. The essential restriction is on the uniform convergence rate for the estimators of the nuisance parameters. In low-dimensional settings, one can consider the local polynomial regression for estimation of the conditional expectations. Under suitable conditions \citep{hansen2008uniform,masry1996multivariate}, the uniform convergence rate of the local polynomial estimators is $(\log n / n)^{2/(d_X+4)}$, which is $o(n^{-1/4})$ if $d_X \leq 3$. In high-dimensional settings, as pointed out by \cite{chernozhukov2018double}, the rate $o(n^{-1/4})$ is often available for common machine learning methods under structured assumptions on the nuisance parameters.\footnote{This includes the LASSO method under sparsity of the nuisance space. See, for example, \cite{buhlmann2011statistics}, \cite{belloni2011l1}, and \cite{belloni2013least}. However, \cite{chernozhukov2018double} also indicate that to prove that machine learning methods achieve the $o(n^{-1/4})$ rate, one will eventually have to use related entropy conditions.} This means that the asymptotic normality of the DML estimator continues to hold. 


Theorem \ref{thm:dml} can be directly used to conduct inference on $\beta$. Confidence regions can be constructed by inverting the usual $t$-tests. These confidence regions are uniformly valid since the convergence results in the above theorem hold uniformly over $\mathcal{P}$. In the next section, we explain why uniform validity is crucial when dealing with weak identification issues.

\section{Weak Identification} \label{sec:weak}

The convergence result established in Theorem \ref{thm:dml} is uniform over the set of DGPs with type probability $p$ bounded away from zero. However, the identification of $\beta$ would be weak in the case where $p$ can be arbitrarily close to zero. This leads to distortion of the uniform size of the test and poor asymptotic approximation in finite-sample settings. This section studies this weak identification issue and proposes an inference procedure that is robust against such a problem.

We begin with a heuristic illustration of the weak identification problem. To ease notation, define $\upsilon = \beta p$ and 
\begin{align*}
	\check{\upsilon} = \check{\beta} \check{p} = \sum_{l=1}^L \sum_{i \in I_l}b \big( \zeta(Z_i,X_i, \check{\pi}^l) \big( \iota (Y_i \mathbf{1}\{T_i=t\}) - \check{Q}^l(X_i) \big) + \check{Q}^l(X_i) \big).
\end{align*}
After a simple calculation, we can write
\begin{align*}
	\check{\beta} - \beta = \frac{\sqrt{n}(\check{\upsilon}-\upsilon) - \beta \sqrt{n}(\check{p}-p)}{\sqrt{n}(\check{p}-p) + \sqrt{n}p} .
\end{align*}
In the above expression, we can interpret the estimation errors $\sqrt{n}(\check{\upsilon}-\upsilon)$ and $\sqrt{n}(\check{p}-p)$ as the noises, while the signal is the term $\sqrt{n}p$. Under the usual asymptotics where $p > 0$ is fixed, the noise terms are bounded in probability, whereas the signal term $\sqrt{n}p \rightarrow \infty$. Hence, the signal dominates the noise, and the estimator $\check{\beta}$ is consistent. However, under asymptotics with a drifting sequence $p = p_n \rightarrow 0$ and $\sqrt{n}p$ converging to a finite constant, the signal and the noise are of the same magnitude, which results in the inconsistency of $\check{\beta}$. This problem is the weak identification issue. In the weak IV literature, a common measure of identification strength is the so-called \textit{concentration parameter}. In our case, the concentration parameter is given by $\sqrt{n}p$ where $\sqrt{n}p \rightarrow \infty$ corresponds to strong identification, and identification is weak when the limit of $\sqrt{n}p$ is finite.

While weak identification is a finite-sample issue, it is formalized using the asymptotic framework. However, the illustration above using asymptotics under drifting sequences is not meant to model DGPs that vary with the sample size $n$. Instead, it is a tool used to detect the lack of uniform convergence. In fact, controlling the uniform size of the test is the key to solving weak identification problems.\footnote{See, for example, \cite{imbens2004confidence}, \cite{mikusheva2007uniform}, and \cite{andrews2020generic}.} Formally, the uniform size of a test is the large sample limit of the supremum of the rejection probability under the null hypothesis, where the supremum is taken over the nuisance parameter space. When testing a null hypothesis on $\beta$ in the GLATE model, the supremum mentioned above is taken over all values of $p>0$. That is, a desirable test should have rejection probability under the null converge to the nominal size uniformly over $p \in (0,1]$. From the previous discussion, we can see that the uniform size can not be controlled using the usual $t$-statistic $\sqrt{n}(\check{\beta} - \beta)/\check{\sigma}$. This failure of uniform convergence, however, does not conflict with Theorem \ref{thm:dml}, where the uniform convergence of $\check{\beta}$ is established only after restricting $p$ to be bounded away from zero.

Inference procedures that are robust against weak identification can be obtained by directly imposing the null hypothesis in the construction of the test statistic. One such example is the well-known Anderson-Rubin (AR) statistic in the weak IV literature. Its idea can be generalized to the GLATE model. We first consider testing the two-sided hypothesis $H_0:\beta = \beta_0$ versus $H_1:\beta \ne \beta_0$. To control the uniform size of the test, we need the test statistic to converge uniformly on the parameter space where (1) $\beta = \beta_0$, and (2) $p$ is allowed to be arbitrarily close to zero. A null-restricted $t$-statistic can be obtained as follows. Notice that when $p>0$, $\beta = \beta_0$ is equivalent to 
\begin{align} \label{eqn:WI-moment-condition}
	0 = \upsilon - \beta_0 p = \mathbb{E} \left[ \psi(Y,T,Z,X,\beta_0,Q,P,\pi) \right].
\end{align}
Its estimate can be written as
\begin{align} \label{eqn:WI-robust}
	\check{\upsilon} - \beta_0 \check{p} = (\check{\upsilon} - \upsilon) - \beta (\check{p} - p) + (\beta - \beta_0) p.
\end{align}
Under the null hypothesis $\beta = \beta_0$, the above estimate does not depend on the concentration parameter $\sqrt{n}p$ and consists only of the noise terms $\check{\upsilon} - \upsilon$ and $\check{p} - p$, whose uniform convergence can be established directly. 

For implementation, this test statistic can be obtained as a straightforward application of the DML procedure described in the previous section to the moment condition (\ref{eqn:WI-moment-condition}).
As a consequence of Proposition \ref{prop:Neyman}, the above moment condition satisfies the Neyman orthogonality condition regardless of the true value of $\beta$. More specifically, the null-restricted $t$-statistic is defined to be 
\begin{align*}
	\check{\rho} = \sqrt{n}(\check{\upsilon} - \beta_0 \check{p}) / \check{\sigma}_\psi,
\end{align*}
where
\begin{align*}
	\check{\sigma}_\psi^2 = \frac{1}{n} \sum_{l = 1}^L \sum_{i \in I_l} \psi(Y_i,T_i,Z_i,X_i,\beta_0,\check{Q}^l,\check{P}^l,\check{\pi}^l)^2.
\end{align*}
The corresponding test of $H_0:\beta = \beta_0$ against $H_1:\beta \ne \beta_0$ rejects for large values of $|\check{\rho}|$.

The same methodology can be applied to testing one-sided hypothesis $H_0:\beta \leq \beta_0$ versus $H_1:\beta > \beta_0$. Under the null hypothesis, $(\beta - \beta_0)p$ is non-positive, suggesting that the test should reject for large values of $\check{\rho}$. Notice that this relies on knowing the sign of $p$ due to the GLATE model structure. This restriction on the sign of $p$ is similar to knowing the first-stage sign in the linear IV model, which is studied by \cite{andrews2017unbiased} in the context of unbiased estimation.

We now define the set of DGPs that allows $p$ to be arbitrarily close to zero.
For any two constants $c_1>c_0>0$, let $\mathcal{P}^{\text{WI}}(c_0,c_1)$ be the set of joint distributions of $(Y,T,Z,X)$ such that 
\begin{enumerate} [label = (\roman*)]
	\item $p \in (0,1]$,
	\item $\mathbb{E}[\psi^2],\pi_z^o(X) \geq c_0,z \in \mathcal{Z}$, and $\abs{Y \mathbf{1}\{T=t\}},|Y \mathbf{1}\{T=t\} - Q_t^o(X)| \leq c_1$.
\end{enumerate}
For any $\beta' \in \mathbb{R}$, let $\mathcal{P}^{\text{WI}}_{\beta'}(c_0,c_1)$ be the subset of $\mathcal{P}^{\text{WI}}(c_0,c_1)$ in which the true value of the parameter $\beta$ is $\beta'$. In particular, $\mathcal{P}^{\text{WI}}_{\beta_0}(c_0,c_1)$ denotes the subset where the null hypothesis is true. The superscript ``WI'' denotes weak identification. The difference between $\mathcal{P}(c_0,c_1)$ and $\mathcal{P}^{\text{WI}}(c_0,c_1)$ is that $\mathcal{P}^{\text{WI}}(c_0,c_1)$ allows the type probability $p$ to be arbitrarily small, whereas the type probabilities in $\mathcal{P}(c_0,c_1)$ are uniformly bounded away from zero. Denote $\mathcal{N}_{\nu}$ as the $\nu$th quantile of the standard normal distribution. The following theorem establishes that the above testing procedures have uniformly correct sizes and are consistent. 

\begin{theorem} \label{thm:weak-id-test}
	Suppose the conditions on the nuisance parameter estimates in Theorem \ref{thm:dml} hold. Let $\alpha \in (0,1)$ be the nominal size of the tests.
	\begin{enumerate} [label = (\roman*)]
		\item The test that rejects $H_0:\beta = \beta_0$ in favor of $H_0:\beta \ne \beta_0$ when $|\check{\rho}| > \mathcal{N}_{1-\frac{\alpha}{2}}$ has (asymptotically) uniformly correct size and is consistent. That is,
		\begin{align*}
			\sup \big\{ \mathbb{P}_P\big(|\check{\rho}| > \mathcal{N}_{1-\frac{\alpha}{2}}\big): P \in \mathcal{P}^{\text{WI}}_{\beta_0}(c_0,c_1) \big\} \rightarrow \alpha
		\end{align*}
		and
		\begin{align*}
			\mathbb{P}_P \big(|\check{\rho}| > \mathcal{N}_{1-\frac{\alpha}{2}} \big) \rightarrow 1, P \in \mathcal{P}^{\text{WI}}_{\beta}(c_0,c_1), \beta \ne \beta_0.
		\end{align*}
		\item The test that rejects $H_0:\beta \leq \beta_0$ in favor of $H_0:\beta > \beta_0$ when $\check{\rho} > \mathcal{N}_{1-\alpha}$ has (asymptotically) uniformly correct size and is consistent. That is,
		\begin{align*}
			\sup \big\{ \mathbb{P}_P\big(\check{\rho} > \mathcal{N}_{1-\alpha} \big): P \in \mathcal{P}^{\text{WI}}_{\beta}(c_0,c_1),\beta \leq \beta_0 \big\} \rightarrow \alpha
		\end{align*}
		and
		\begin{align*}
			\mathbb{P}_P \big(\check{\rho} > \mathcal{N}_{1-\alpha} \big) \rightarrow 1, P \in \mathcal{P}^{\text{WI}}_{\beta}(c_0,c_1), \beta > \beta_0.
		\end{align*}
	\end{enumerate}
\end{theorem}

\section{Empirical Application} \label{sec:empirical}

In this section, we apply the theoretical results to data from the Oregon Health Insurance Experiment \cite{finkelstein2012oregon} and examine the effects on the health of different sources of health insurance. The experiment is conducted by the state of Oregon between March and September 2008. A series of lottery draws were administered to award the participants the option of enrolling in the Oregon Health Plan Standard, which is a Medicaid expansion program available for Oregon adult residents that have limited income. Follow-up surveys were sent out in several waves to record, among many variables, the participants' insurance plan and health status. \cite{finkelstein2012oregon} obtain the effects of insurance coverage by using a LATE model. We apply the GLATE model can study the effect heterogeneity across different sources of insurance.

According to the data, many lottery winners did not choose to participate in the Medicaid program. Instead, they went with other insurance plans or chose not to have any health insurance. Based on this observation, we can set up the GLATE model. The instrument $Z$ is the binary lottery that determines whether an individual is selected. The covariates $X$ include the number of household members and survey waves. Given $X$, $Z$ is randomly assigned \citep[][p1071]{finkelstein2012oregon}.\footnote{Though the covariates are discrete, the methods developed in this paper are still different from linear regressions in \cite{finkelstein2012oregon}.} The treatment $T$ is the insurance plan, which contains three categories: Medicaid ($m$), non-Medicaid insurance plans ($nm$), and no health insurance ($no$). The second category includes Medicare, private plans, employer plans, and other plans. The counterfactual health plan choices under different lottery results are the variables $T_0$ and $T_1$. The unordered monotonicity condition requires that any participant who changes insurance plan due to winning the lottery does so to enroll in the Medicaid program.

The above setup is the same as Example \ref{eg:3t2z}, with types. We follow the terminologies in \cite{kline2016evaluating} and define the following six type sets by their counterfactual insurance plan choices:
\begin{enumerate}
	\item $no$-never takers: $S \in \Sigma_{no,2} = \{ s_1 \}$, $T_{0} = T_1 = no$;
	\item $nm$-never takers: $S \in \Sigma_{nm,2} = \{ s_2 \}$, $T_{0} = T_1 = nm$;
	\item always takers: $S \in \Sigma_{m,2} = \{ s_3 \}$, $T_{0} = T_1 = m$;
	\item $no$-compliers: $S \in \Sigma_{no,1} = \{ s_4 \}$, $T_0 = no$, $T_1 = m$;
	\item $nm$-compliers: $S \in \Sigma_{nm,1} = \{ s_5 \}$, $T_0 = nm$, $T_1 = m$;
	\item compliers: $S \in \Sigma_{m,1} = \{ s_4, s_5 \}$, $T_{0} \ne m$, $T_1 = m$.
\end{enumerate}
The two groups of never takers choose not to join Medicaid regardless of the offer. Always takers manage to enroll in Medicaid even without an offer. The $no$- and $nm$- compliers switch to Medicaid from no insurance plan and other plans, respectively, upon winning the lottery. Combining these two groups gives the larger set of compliers. 

Table \ref{tb:app-prob} shows the estimated probabilities of the six types.\footnote{We use the data from the 12-month survey. After taking care of the missing values, we are left with $23290$ observations. For cross-fitting, we choose $L=10$.} We can see that half of the population are $no$-never takers, who are never covered by any insurance plan. The compliers make up around one-fifth of the population. There are effectively no $nm$-compliers, meaning that the experiment does not crowd out other insurance plan choices. These findings are consistent with \cite{finkelstein2012oregon}.

\begin{table}[!htbp] 
    \centering
    \begin{tabular}{*3c}
    \toprule
    Type & Probability & Estimate (se) \\ \midrule
	$no$-never takers & $p_{no,2}$ & .492 (.046) \\
	$nm$-never takers & $p_{nm,2}$ & .208 (.018) \\
	always takers & $p_{m,2}$ & .116 (.018) \\
	$no$-compliers & $p_{no,1}$ & .197 (.059) \\
	$nm$-compliers & $p_{nm,1}$ & .010 (.024) \\
	compliers & $p_{m,1}$ & .208 (.060) \\
    \bottomrule
\end{tabular}
\caption{Estimated probability of different types.}
\label{tb:app-prob}
\end{table}

The outcome of interest $Y$ is health status, which is (inversely) measured by the number of days (out of past 30) when poor health impaired regular activities.\footnote{Other types of outcomes are also studied by \cite{finkelstein2012oregon}, including health care utilization and financial strain. Here we only focus on health status for simplicity.} The potential outcomes are denoted by $Y_{no}$, $Y_{nm}$, and $Y_{m}$. By Theorem \ref{thm:id1}, we can identify the distribution of $Y_{no}$ for $no$-never takers and $no$-compliers, the distribution of $Y_{nm}$ for $nm$-never takers and $nm$-compliers, and the distribution of $Y_{nm}$ for always takers and compliers. Table \ref{tb:app-lasf} reports the estimated LASFs.\footnote{The LASF $\beta_{nm,1}$ is excluded because there are few $nm$-compliers as reported in Table \ref{tb:app-prob}.} We can clearly see a pattern of self-selection into the treatment. For example, when there is no insurance coverage, the potential health status of $no$-compliers is worse than $no$-never takers and therefore choose to enroll in Medicaid. 

\begin{table}[!htbp] 
    \centering
    \begin{tabular}{*4c}
    \toprule
    Type & Treatment & LASF & Estimate (se) \\ \midrule
	$no$-never takers & $no$ & $\beta_{no,2}$ & 6.78 (1.19) \\
	$nm$-never takers & $nm$ & $\beta_{nm,2}$ & 7.74 (1.05) \\
	always takers & $m$ & $\beta_{m,2}$ & 9.96 (1.75) \\
	$no$-compliers & $no$ & $\beta_{no,1}$ & 11.50 (2.92) \\
	compliers & $m$ & $\beta_{m,1}$ & 0.48 (3.42) \\
    \bottomrule
\end{tabular}
\caption{Estimated LASFs.}
\label{tb:app-lasf}
\end{table}

\section{Concluding Remarks} \label{sec:conclusion}


In this paper, we considered the estimation of the causal parameters, LASF and LASF-T, in the GLATE model by using the EIF. The proposed DML estimator satisfies the SPEB and can be applied in situations, such as high-dimensional settings, where Donsker properties fail. For inference, we proposed generalized AR tests robust against weak identification issues. 
Currently, empirical researchers use the TSLS and control the covariates linearly in models with multi-valued treatments and instruments. This linear specification does not have LATE interpretation, as pointed out by \cite{blandhol2022tsls}. Therefore, we advocate using the semiparametric methods studied by this paper in those cases.

\bigskip
\begin{center}
{\large\bf SUPPLEMENTARY MATERIAL}
\end{center}

\appendix
\numberwithin{equation}{section}

\section{Technical Proofs} \label{app:proof}
\label{appendix:pf}

In this section, we prove the theorems and propositions stated in the main text. We assume that Assumptions \ref{ass:CI} and \ref{ass:UM} hold throughout this section.

\subsection{Proof of the Identification Results}

\begin{lemma} \label{lm:CI}
	$S \perp Z \mid X$ and $t \in \mathcal{T}$, $Y_t \perp T \mid S,X$.
\end{lemma}
\begin{proof} [Proof of Lemma \ref{lm:CI}]
	The first statement follows from the definition of $S$ and the fact that $Z$ is independent of the vector $( T_{z_1}, \cdots, T_{z_{N_Z}} )$ conditioning on $X$. For the second statement, $T$ is entirely determined by $\left(S,Z,X \right) $. Hence, given $S$ and $X$, $T$ is independent of $Y_t$ since $Z$ is independent of $( Y_{t_1}, \cdots, Y_{t_{N_T}} ) $ conditional on $X$.
\end{proof}

\begin{lemma} \label{lm:HP6}
	For each $t \in \mathcal{T}$ and $k = 1, \cdots, N_Z$, the following identification results hold.
	\begin{enumerate}[label = (\roman*)]
		\item $\mathbb{P}(S \in \Sigma_{t,k} \mid X) = b_{t,k}P_t(X)$ a.s.
		\item $\mathbb{E} \left[ Y_t \mid S \in \Sigma_{t,k}, X\right] = (b_{t,k}  Q_t(X)) / (b_{t,k} P_t(X))$ a.s.
	\end{enumerate}
\end{lemma}
\begin{proof} [Proof of Lemma \ref{lm:HP6}]
	This is Theorem T-6 in \cite{heckman2018unordered}. The conditioning is explicitly presented.
\end{proof}

\begin{proof}[Proof of Theorem \ref{thm:id1}]
	The first statement follows from applying the law of iterated expectation to Lemma \ref{lm:HP6}(i). For the second statement, we can apply Bayes rule to Lemma \ref{lm:HP6} and obtain that
		\begin{align*}
			\mathbb{E} \left[ Y_t \mid S \in \Sigma_{t,k} \right] &= \int \mathbb{E} \left[ Y_t \mid S \in \Sigma_{t,k}, X=x \right] f_{X\mid S \in \Sigma_{t,k}}(x)dx \\
			& = \int \mathbb{E} \left[ Y_t \mid S \in \Sigma_{t,k}, X=x \right] \frac{\mathbb{P}(S \in \Sigma_{t,k} \mid X=x)}{\mathbb{P}(S \in \Sigma_{t,k})} f_X(x)dx \\
			& =  \mathbb{E} \left[ b_{t,k} Q_t(X) \right] / p_{t,k} ,
		\end{align*}
		where $f_{X\mid S \in \Sigma_{t,k}}$ denotes the conditional density function of $X$ given type $S \in \Sigma_{t,k}$.
\end{proof}

\begin{proof}[Proof of Theorem \ref{thm:id2}]
	By Lemma L-16 of \cite{heckman2018web}, we know that under the unordered monotonicity assumption, $B_t[\cdot,i] = B_t[\cdot,i']$ for all $s_i,s_{i'} \in \Sigma_{t,k}$. Thus, the set $\mathcal{Z}_{t,k}$ always exists.
	For the first statement, we have
		\begin{align*} 
			\mathbb{P}\left( T = t, S \in \Sigma_{t,k}  \right) & = \mathbb{P}\left( Z \in \mathcal{Z}_{t,k}, S \in \Sigma_{t,k}  \right) \\
			& = \mathbb{E} \left[ \mathbb{P}\left( Z \in \mathcal{Z}_{t,k}, S \in \Sigma_{t,k} \mid X \right) \right] \\
			& = \mathbb{E} \left[ \mathbb{P}\left( Z \in \mathcal{Z}_{t,k} \mid X \right) \mathbb{P}\left(  S \in \Sigma_{t,k} \mid X \right) \right] \\
			& = \mathbb{E} \left[ b_{t,k} P_t(X)\pi_{{t,k}}(X)  \right],
		\end{align*}
		where the second equality follows from the law of iterated expectations and the third equality follows from the fact that $Z \perp S \mid X$ (Lemma \ref{lm:CI}).
		For the second statement, notice that 
		\begin{align*}
			\mathbb{P}(T=t,S \in \Sigma_{t,k} \mid X=x) & = \mathbb{P}(T=t \mid S \in \Sigma_{t,k} , X=x) \mathbb{P}(S \in \Sigma_{t,k} \mid X=x) \\
			& = \mathbb{P}(Z \in \mathcal{Z}_{t,k} \mid X) \mathbb{P}(S \in \Sigma_{t,k} \mid X=x) \\
			& = \pi_{t,k}(X) b_{t,k}P_t(X).
		\end{align*}
		By Lemma \ref{lm:CI}, we know that
		\begin{align*}
			\mathbb{E} \left[ Y_t \mid T=t, S \in \Sigma_{t,k}, X=x \right] = \mathbb{E} \left[ Y_t \mid S \in \Sigma_{t,k}, X=x \right].
		\end{align*}
		Therefore, we can apply Bayes rule and obtain that 
		\begin{align*}
			&\mathbb{E} \left[ Y_t \mid T=t, S \in \Sigma_{t,k} \right] \\
			=& \int \mathbb{E} \left[ Y_t \mid T=t,S \in \Sigma_{t,k}, X=x \right] f_{X\mid T=t, S \in \Sigma_{t,k}}(x)dx \\
			=& \int \mathbb{E} \left[ Y_t \mid S \in \Sigma_{t,k}, X=x \right] \frac{P(T=t,S \in \Sigma_{t,k} \mid X=x)}{P(T=t,S \in \Sigma_{t,k})} f_X(x)dx \\
			=& \int \frac{b_{t,k} Q_t(X)}{b_{t,k}P_t(X)} \times \frac{\pi_{t,k}(X) b_{t,k}P_t(X)}{q_{t,k}} f_X(x)dx \\ 
			=& \mathbb{E} \left[  b_{t,k} Q_t(X)\pi_{t,k}(X) \right] / q_{t,k}.
		\end{align*}

\end{proof}

\subsection{Semiparametric Efficiency Calculations}

We follow the method developed by \cite{newey1990semiparametric}. The likelihood of the GLATE model can be specified as
\begin{equation*}
	\mathcal{L}\left( Y,T,Z,X \right) 
	= f_X(X) \prod_{z \in \mathcal{Z}} \Big( f_{z}(Y,T \mid X) \pi_{z}(X) \Big)^{\mathbf{1}\{Z = z\}} ,
\end{equation*}
where $f_z(\cdot, \cdot \mid X)$ denotes the conditional density of  $Y,T$ given $Z =z$ and $X$. In a regular parametric submodel, where the true underlying probability measure $P$ is indexed by $\theta^o$, we use the following notations to represent the score functions:
\begin{align*}
	& s_z(Y, Z \mid X;\theta) = \frac{\partial }{ \partial \theta} \log \left( f_z(Y, T \mid X ; \theta) \right) ,\\
	& s_{\pi}(Z \mid X; \theta) = \sum_{z \in \mathcal{Z}}\mathbf{1}\{Z = z\} \frac{\partial }{ \partial \theta} \log \left( \pi_{z}(X ; \theta) \right)  ,\\
	& s_X(X;\theta) = \frac{\partial }{ \partial \theta} \log \left( f_X(X ; \theta) \right) .
\end{align*}
	The score in a regular parametric submodel is
	\begin{align*}
		s_{\theta^o}(Y,T,Z,X) = \sum_{z \in \mathcal{Z}} \mathbf{1}\{Z = z\} s_{z}\left( Y,T \mid X; \theta^o \right)  + s_{\pi}(Z \mid X; \theta^o)  + s_X(X; \theta^o).
	\end{align*}
	Hence, the tangent space of the model is 
	\begin{align*}
		\mathscr{S} & = \big\{ s \in L^2_0: s(Y,T,Z,X) =  \sum_{z \in \mathcal{Z}} \mathbf{1}\{Z = z\} s_{z}\left( Y,T \mid X\right)  + s_{\pi}(Z \mid X)  + s_X(X) \\ 
		& \quad \quad \text{ for some } s_{z}, s_{\pi}, s_X \text{ such that } \int s_z(y,t \mid X ) f_z(y,t \mid X) dydt \equiv 0, \forall z; \\
		& \quad \quad \sum_{z \in \mathcal{Z}} s_{\pi}(z \mid X) \pi_z(X) \equiv 0 \text{, and } \int s_X(x) f_X(x) dx = 0 \big\},
	\end{align*}
	where $L^2_0$ is a subspace of $L^2$ that contains the mean zero functions.

\begin{proof} [Proof of Theorem \ref{thm:speb1}] 
	We only prove statements (i) and (ii) since (iii) and (iv) are easier cases that can be proved along the way. We start with the first statement. The path-wise differentiability of the parameter $\beta_{t,k}$ can be verified in the following way: in any parametric submodel, we have
	\begin{align*}
			\frac{\partial }{ \partial \theta} \beta_{t,k}(\theta) \Big|_{\theta = \theta^o} & =  \frac{\partial }{\partial \theta}( b_{t,k} \mathbb{E}_\theta \left[  Q_t(X) \right]/p_{t,k} ) \big|_{\theta = \theta^o} \\
			& = \frac{1}{p_{t,k}} \left( (\partial b_{t,k} \mathbb{E}_\theta \left[  Q_t(X) \right]/\partial \theta) |_{\theta = \theta^o} - (b_{t,k} \mathbb{E}_\theta \left[  Q_t(X) \right] / p_{t,k}) (\partial p_{t,k}/\partial \theta) |_{\theta = \theta^o} \right) \\
			& = \frac{1}{p_{t,k}} b_{t,k}  \left( \frac{\partial }{\partial \theta} \mathbb{E}_\theta \left[  Q_t(X) \right] \big|_{\theta = \theta^o} -  \frac{\partial }{\partial \theta} \mathbb{E}_\theta \left[  P_t(X) \right] \big|_{\theta = \theta^o} \beta_{t,k} \right),
		\end{align*}
		where $\frac{\partial }{\partial \theta} \mathbb{E}_\theta \left[  Q_t(X) \right] |_{\theta = \theta^o}$ and $\frac{\partial }{\partial \theta} \mathbb{E}_\theta \left[  P_t(X) \right] |_{\theta = \theta^o}$ are $N_Z \times 1$ random vectors whose typical element can be represented respectively by
		\begin{align*}
			& \int y \mathbf{1}\{\tau = t\} s_{z}(y,\tau \mid x; \theta^o) f_{z}(y,\tau \mid x ; \theta^o) f_X(x ; \theta^o) dyd\tau dx \\
			+ &  \int y \mathbf{1}\{\tau = t\} s_X( x; \theta^o) f_{z}(y,\tau \mid x ; \theta^o) f_X(x ; \theta^o) dyd\tau dx
		\end{align*}
		and
		\begin{align*}
			& \int \mathbf{1}\{\tau = t\} s_{z}(y,\tau \mid x; \theta^o) f_{z}(y,\tau \mid x ; \theta^o) f_X(x ; \theta^o) dyd\tau dx \\
			+ &  \int \mathbf{1}\{\tau = t\} s_X( x; \theta^o) f_{z}(y,\tau \mid x ; \theta^o) f_X(x ; \theta^o) dyd\tau dx,
		\end{align*}
		respectively, for $z \in \mathcal{Z}$. The EIF is characterized by the condition that
		\begin{align*}
			\frac{\partial }{ \partial \theta} \beta_{t,k}(\theta) \Big|_{\theta = \theta^o} = \mathbb{E} \left[ \psi_{\beta_{t,k}} s_{\theta^o} \right] \text{, and } \psi_{\beta_{t,k}} \in \mathscr{S}.
		\end{align*}
		The expression of $\psi_{\beta_{t,k}}$ given in Equation (\ref{eqn:Psi_beta}) meets the above requirements. In particular, the correspondence between terms in the EIF and path-wise derivative appears exactly as in Lemma 1 of \cite{hong2010supplement}.

		For the second statement,
		the path-wise derivative of $\gamma_{t,k}$ can be computed similarly. 
		\begin{align*}
			\frac{\partial }{ \partial \theta} \gamma_{t,k}(\theta) \Big|_{\theta  = \theta^o} & = \frac{1}{q_{t,k}} b_{t,k}   \frac{\partial }{\partial \theta} \mathbb{E}_\theta \left[  Q_t(X) \pi_{{t,k}}(X) \right] \Big|_{\theta = \theta^o} \\
			& -  \frac{\gamma_{t,k}}{q_{t,k}} b_{t,k}  \frac{\partial }{\partial \theta} \mathbb{E}_\theta \left[  P_t(X) \pi_{{t,k}}(X) \right] \Big|_{\theta = \theta^o} ,
		\end{align*}
		where $\frac{\partial }{\partial \theta} \mathbb{E}_\theta [  Q_t(X) \pi_{W_{t,k}}(X) ] |_{\theta = \theta^o}$ and $\frac{\partial }{\partial \theta} \mathbb{E}_\theta [  P_t(X) \pi_{W_{t,k}}(X) ] |_{\theta = \theta^o}$ are $N_Z \times 1$ random vectors whose typical element can be represented by
		\begin{align*}
			& \int y \mathbf{1}\{\tau = t\} s_{z}(y,\tau \mid x; \theta^o) \pi_{W_{t,k}}(x;\theta^o) f_{z}(y,\tau \mid x ; \theta^o) f_X(x ; \theta^o) dyd\tau dx \\
			+ &  \int y \mathbf{1}\{\tau = t\} s_X( x; \theta^o) \pi_{W_{t,k}}(x;\theta^o) f_{z}(y,\tau \mid x ; \theta^o) f_X(x ; \theta^o) dyd\tau dx \\
			+ & \int y \mathbf{1}\{\tau = t\}  \left(  \frac{\partial }{ \partial \theta} \pi_{{t,k}}(X ; \theta) \big|_{\theta = \theta^o} \right) f_{z}(y,\tau \mid x;\theta^o ) f_X(x;\theta^o) dyd\tau dx ,
		\end{align*}
		and
		\begin{align*}
			& \int \mathbf{1}\{\tau = t\} s_{z}(y,\tau \mid x; \theta^o) \pi_{W_{t,k}}(x;\theta^o) f_{z}(y,\tau \mid x ; \theta^o) f_X(x ; \theta^o) dyd\tau dx \\
			+ &  \int \mathbf{1}\{\tau = t\} s_X( x; \theta^o) \pi_{W_{t,k}}(x;\theta^o) f_{z}(y,\tau \mid x ; \theta^o) f_X(x ; \theta^o) dyd\tau dx \\
			+ & \int \mathbf{1}\{\tau = t\}  \left(  \frac{\partial }{ \partial \theta} \pi_{{t,k}}(X ; \theta) \big|_{\theta = \theta^o} \right) f_{z}(y,\tau \mid x;\theta^o ) f_X(x;\theta^o) dyd\tau dx ,
		\end{align*}
		respectively, for $z \in \mathcal{Z}$. The main difference appears when dealing with the last terms in the above two expressions, which can be matched with terms in the efficient influence function of the following two forms
		\begin{align*}
			& \mathbb{E} \left[ Y \mathbf{1}\{T=t\} \mid Z=z, X \right] \left( \mathbf{1}\{Z \in \mathcal{Z}_{t,k}\} - \pi_{{t,k}}(X)  \right), \text{ and } \\ 
			& \mathbb{E} \left[ \mathbf{1}\{T=t\} \mid Z=z,X \right] \left( \mathbf{1}\{Z \in \mathcal{Z}_{t,k}\} - \pi_{{t,k}}(X)  \right).
		\end{align*}
		Take the latter one as an example. Notice that
		\begin{align*}
			\mathbf{1} \{Z \in \mathcal{Z}_{t,k}\} - \pi_{{t,k}}(X) = \sum_{z \in \mathcal{Z}_{t,k}} \left( \mathbf{1} \{Z =z\} - \pi_{z}(X) \right), 
		\end{align*}
		and
		\begin{align*}
			\left( \mathbf{1} \{Z =z\} - \pi_{z}(X) \right) s_\pi(Z \mid X ; \theta^o) = \frac{\mathbf{1} \{Z =z\}}{\pi_z(X)} \frac{\partial }{ \partial \theta} \pi_{z}(X ; \theta) \big|_{\theta = \theta^o} - \pi_z(X) s_\pi(Z \mid X; \theta^o).
		\end{align*}
		By the law of iterated expectation, we have
		\begin{align*}
			&  \mathbb{E} \left[ \mathbb{E} \left[ \mathbf{1}\{T=t\} \mid Z=z,X \right] \left( \mathbf{1} \{Z =z\} - \pi_{z}(X) \right) s_\pi(Z \mid X; \theta^o) \right] \\
			= &   \mathbb{E} \left[ \mathbb{E} \left[ \mathbf{1}\{T=t\} \mid Z=z,X \right] \mathbb{E} \left[ \mathbf{1} \{Z =z\}/\pi_z(X) \mid X \right] \frac{\partial }{ \partial \theta} \pi_{z}(X ; \theta) \big|_{\theta = \theta^o} \right] \\
			- & \mathbb{E} \left[ \mathbb{E} \left[ \mathbf{1}\{T=t\} \mid Z=z,X \right] \pi_{z}(X) \mathbb{E} \left[ s_\pi(Z \mid X; \theta^o) \mid X \right] \right] \\
			= & \mathbb{E} \left[ \mathbb{E} \left[ \mathbf{1}\{T=t\} \mid Z=z,X \right] \frac{\partial }{ \partial \theta} \pi_{z}(X ; \theta) \big|_{\theta = \theta^o} \right] \\
			= & \int \mathbf{1}\{\tau = t\}  \left(  \frac{\partial }{ \partial \theta} \pi_{z}(X ; \theta) \big|_{\theta = \theta^o} \right) f_{z}(y,\tau \mid x;\theta^o ) f_X(x;\theta^o) dyd\tau dx.
		\end{align*}

\end{proof}

\begin{proof} [Proof of Proposition \ref{prop:cep}]
	This proof is based on Section 4 in \cite{newey1994asymptotic}. We focus on the case of $\beta_{t,k}$. The other cases are similar. To ease notation, let $h_t = \left( h_{Y,t,Z},h_{t,Z},\pi \right)' $.
	The estimator $\hat{\beta}_{t,k}$ is defined by the moment condition 
	\begin{align*}
		\mathbb{E}[M\left(X, \beta_{t,k}, h_t\right)] = 0,
	\end{align*}
	where 
	\begin{align*}
		M\left(X, \beta_{t,k}, h_t\right) \equiv b_{t,k}  \left( \frac{h_{Y,t,z_{1}}(X)}{\pi_{z_1}(X)}, \cdots, \frac{h_{Y,t,z_{N_Z}}(X)}{\pi_{z_{N_Z}}(X)} \right)' - \beta_{t,k} b_{t,k} \left( \frac{h_{t,z_{1}}(X)}{\pi_{z_1}(X)}, \cdots, \frac{h_{t,z_{N_Z}}(X)}{\pi_{z_{N_Z}}(X)} \right)'.
	\end{align*}
	We then compute the derivatives of $M$ with respect to the parameters:
	\begin{align*}
		\mathbb{E} \left[\partial M / \partial \beta_{t,k}   \right] & =- b_{t,k} \mathbb{E} \left[ P_t(X) \right] =  - p^o_{t,k} \\
		\partial M / \partial h_{Y,t,z_i}  |_{h_t = h_t^o} & = b_{t,k}[i] / \pi^o_{z_i}(X)   \equiv \delta_{Y,t,z_i}(X) \\
		\partial / \partial h_{t,z_i} M |_{h_t = h_t^o} & = - (\beta_{t,k}b_{t,k}[i]) / \pi^o_{z_i}(X) \equiv \delta_{t,z_i}(X) \\
		\partial M / \partial \pi_{z_i}   |_{h_t = h_t^o} & = - (b_{t,k}[i] Q^o_{t,z_i}(X)) / \pi^o_{z_i}(X)  + (\beta_{t,k}b_{t,k}[i] P^o_{t,z_i}(X))/ \pi^o_{z_i}(X)  \equiv \delta_{\pi,z_i}(X),
	\end{align*}
	where $b_{t,k}[i]$ denotes the $i$th element of the vector $b_{t,k}$. Define
	\begin{align*}
		\alpha\left( Y,T,Z,X \right)  & \equiv \sum_{z \in \mathcal{Z}} \delta_{Y,t,z}(X) \left( \mathbf{1}\{Z = z\}  Y  \mathbf{1}\{T = t\} - h^o_{Y,t,z}(X) \right) \\
		&  \quad + \sum_{z \in \mathcal{Z}} \delta_{t,z}(X) \left( \mathbf{1}\{Z = z\}  \mathbf{1}\{T = t\} - h^o_{t,z}(X) \right) \\
		& \quad + \sum_{z \in \mathcal{Z}} \delta_{\pi,z}(X) \left( \mathbf{1}\{Z = z\}  - \pi^o_{z}(X) \right) .
	\end{align*}
	We have
	\begin{align*}
		\alpha\left( Y,T,Z,X \right) & = b_{t,k}  \zeta(Z,X,\pi^o) \left(  \iota (Y  \mathbf{1}\{T=t\}) - Q^o_t(X) \right) \\
		& \quad - \beta_{t,k}^o b_{t,k} \zeta(Z,X,\pi^o)   \left( \iota \mathbf{1}\{T=t\} - P^o_t(X) \right).
	\end{align*}
	Then \citeauthor{newey1994asymptotic}'s (\citeyear{newey1994asymptotic}) Proposition 4 suggests that the influence function of the estimator $\hat{\beta}_{t,k}$ is $(M+\alpha) / p_{t,k}$ which is equal to the EIF $\psi^{\beta_{t,k}}$.

\end{proof} 

\subsection{Proof of Robustness Results}

\begin{proof} [Proof of Proposition \ref{prop:DR}]
	We prove the case for $\psi^{p_{t,k}}$, the other cases can be dealt with analogously. First assume $\pi = \pi^o$, then 
	\begin{align*}
		\mathbb{E} \left[ \mathbf{1}\{Z=z\} / \pi_z^o(X) \mid X \right] = 1,
	\end{align*}
	which implies that $\mathbb{E}\left[ \zeta(Z,X, \pi^o) \mid X \right]$ is almost surly equal to the identity matrix $\mathbf{I}$. By the law of total expectations, we have
	\begin{align*}
		\mathbb{E} \left[ \mathbf{1}\{T=t\} \mathbf{1}\{Z=z\} / \pi_z^o(X) \mid X \right] =\mathbb{E} \left[ \mathbf{1}\{T=t\} \mid Z=z, X \right] = P_{t,z}^o(X),
	\end{align*}
	which implies that $\mathbb{E} \left[\zeta(Z,X, \pi^o) \iota \mathbf{1}\{T=t\}  \right] = \mathbb{E} \left[P_t^o(X)  \right]$. Therefore,
	\begin{align*}
		&b_{t,k}\mathbb{E}[ \zeta(Z,X, \pi^o) \left( \iota (\mathbf{1}\{T=t\}) - P_t(X) \right) + P_t(X)  ] \\
		=& b_{t,k}\mathbb{E} \left[\zeta(Z,X, \pi^o) \iota \mathbf{1}\{T=t\}  \right] + b_{t,k}\mathbb{E}\left[ (\mathbf{I} - \zeta(Z,X, \pi^o)) P_t(X)\right] = b_{t,k}\mathbb{E} \left[P_t^o(X)  \right] =  p_{t,k}^o.
	\end{align*}
	Now suppose that $P_t = P_t^o$. Then by the law of total expectation, we have
	\begin{align*}
		& \mathbb{E} [ \mathbf{1}\{Z=z\}(\mathbf{1}\{T=t\} - P_{t,z}^o(X)) \mid X ] \\
		=&  \pi_z(X) \mathbb{E} [ \mathbb{E}[\mathbf{1}\{T=t\}\mid Z=z,X] - P_{t,z}^o(X) \mid X ] = 0.
	\end{align*}
	This implies that $\mathbb{E} [ \zeta(Z,X, \pi) ( \iota (\mathbf{1}\{T=t\}) - P_{t}^o(X) )  ] = 0$. Hence,
	\begin{align*}
		b_{t,k}\mathbb{E} \left[ \zeta(Z,X, \pi) \left( \iota (\mathbf{1}\{T=t\}) - P_t^o(X) \right) + P_t^o(X)\right] = b_{t,k}\mathbb{E} \left[P_t^o(X)  \right]
		= p_{t,k}^o .
	\end{align*}
	This proves the proposition.
\end{proof}

\begin{proof} [Proof of Proposition \ref{prop:Neyman}]
	Since $b_{t,k}$ is a finite vector, it suffices to verify the Neyman orthogonality condition for $\psi_z$, which is defined by
	\begin{align*}
		& \psi_z(Y,T,Z,X,\beta_{t,k},Q_t,P_t,\pi_z) \\
		\equiv & \big( (\mathbf{1}\{Z = z\} / \pi_z(X)) \left( \mathbf{1}\{T=t\} - P_{t,z}(X) \right) + P_{t,z}(X) \big) \beta_{t,k} \\
		& - (\mathbf{1}\{Z = z\} /\pi_z(X)) \left( Y \mathbf{1}\{T=t\} - Q_{t,z}(X) \right) - Q_{t,z}(X).
	\end{align*}
	We want to show that 
	\begin{align*}
		\frac{d}{dr} \mathbb{E} \left[ \psi_z(Y,T,Z,X,\beta_{t,k},Q_t^r,P_t^r,\pi_z^r) \right] \Big|_{r=0} = 0,
	\end{align*}
	where $Q_t^r = Q_t^o + r(Q_t - Q_t^o),$ $P_t^r = P_t^o + r(P_t - P_t^o),$ and $\pi_z^r = \pi_z^o + r(\pi_z - \pi_z^o)$. In fact,
	\begin{align*}
		&\frac{d}{dr} \mathbb{E} \left[ \psi_z(Y,T,Z,X,\beta_{t,k},Q_t^r,P_t^r,\pi_z^r) \right] \big|_{r=0} \\
		=& \mathbb{E} \Big[  \frac{ -\mathbf{1}\{Z = z\} }{(\pi^r_z(X))^2} \left( \mathbf{1}\{T=t\} - P^r_{t,z}(X) \right) \left( \pi_z(X) - \pi^o_z(X) \right) \beta_{t,k} \\
		& + \left(  P_{t,z}(X) - P^o_{t,z}(X) - \frac{\mathbf{1}\{Z = z\} }{\pi^r_z(X)} \left( P_{t,z}(X) - P^o_{t,z}(X) \right)\right) \beta_{t,k} \\
		& +  \frac{ \mathbf{1}\{Z = z\} }{(\pi^r_z(X))^2} \left( Y \mathbf{1}\{T=t\} - Q^r_{t,z}(X) \right) \left( \pi_z(X) - \pi^o_z(X) \right)  \\
		&  - (Q_{t,z}(X) - Q^o_t(X)) + \frac{\mathbf{1}\{Z = z\} }{\pi^r_z(X)} \left( Q_{t,z}(X) - Q^o_{t,z}(X) \right) \Big] \Big|_{r=0} \\
		 =& \mathbb{E} \Big[  \frac{ -\mathbf{1}\{Z = z\} }{(\pi^o_z(X))^2} \left( \mathbf{1}\{T=t\} - P^o_{t,z}(X) \right) \left( \pi_z(X) - \pi^o_z(X) \right) \beta_{t,k} \\
		 & + \left(  P_{t,z}(X) - P^o_{t,z}(X) - \frac{\mathbf{1}\{Z = z\} }{\pi^o_z(X)} \left( P_{t,z}(X) - P^o_{t,z}(X) \right)\right) \beta_{t,k} \\
		 & + \frac{ \mathbf{1}\{Z = z\} }{(\pi^o_z(X))^2} \left( Y \mathbf{1}\{T=t\} - Q^o_{t,z}(X) \right) \left( \pi_z(X) - \pi^o_z(X) \right)  \\
		 &   - (Q_{t,z}(X) - Q^o_{t,z}(X)) + \frac{\mathbf{1}\{Z = z\} }{\pi^o_z(X)} \left( Q_{t,z}(X) - Q^o_{t,z}(X) \right) \Big] ,
	\end{align*}
	which equals zero because of the following three identities:
	\begin{align*}
		\mathbb{E} [ \mathbf{1}\{Z = z\} / \pi^o_z(X) \mid X ] = 1, \\
		\mathbb{E} [ \mathbf{1}\{Z = z\} / \pi^o_z(X) ( \mathbf{1}\{T=t\} - P^o_{t,z}(X) ) \mid X ] = 0, \\
		\mathbb{E} [ \mathbf{1}\{Z = z\} / \pi^o_z(X) ( Y\mathbf{1}\{T=t\} - Q^o_{t,z}(X) ) \mid X ] =0.
	\end{align*}

\end{proof}

\begin{proof} [Proof of Theorem \ref{thm:dml}]
	The asserted claims follow from Theorem 3.1, Theorem 3.2, and Corollary 3.2 of \cite{chernozhukov2018double} (henceforth referred to as the DML paper). We want to verify their Assumption 3.1 and 3.2. Adopting the notation from the DML paper, we let
	\begin{align*}
		\psi^a(T,Z,X,P_t,\pi) = -b_{t,k} \left( \zeta(Z,X, \pi) \left( \iota \mathbf{1}\{T=t\} - P_t(X) \right) + P_t(X) \right)
	\end{align*}
	and
	\begin{align*}
		\psi^b(Y,T,Z,X,Q_t,\pi) = b_{t,k} \left( \zeta(Z,X, \pi) \left( \iota (Y\mathbf{1}\{T=t\}) - Q_t(X) \right) + Q_t(X) \right)
	\end{align*}
	so that the linearity of the moment condition (with respect to $\beta_{t,k}$) is verified by the fact that $\psi = \psi^a \beta^{t,k} + \psi^b$.
	Define\footnote{For simplicity, we drop the superscript $l$ in the nonparametric estimators.}
	\begin{align*}
		\epsilon_n = \max_{z \in \mathcal{Z}} \big( \lVert \hat{Q}_{t,z} - Q_{t,z}^o \rVert_2 \vee \lVert \hat{P}_{t,z} - P_t^o \rVert_2 \vee  \norm{\hat{\pi}_{z} - \pi_{z}^o}_2 \big) .
	\end{align*}
	By assumption on the convergence rates of the nonparametric estimators, we have $\epsilon_n = o(n^{-1/4})$. Define $C_\epsilon = C_{\epsilon,1} \vee C_{\epsilon,2} \vee C_{\epsilon,3} \vee C_{\epsilon,4} $, where $C_{\epsilon,1},C_{\epsilon,2},C_{\epsilon,3},$ and $C_{\epsilon,4}$ are positive constant that only depends on $C$ and $\epsilon$ and are specified later in the proof. Let $\delta_n$ be a sequence of positive constants approaching zero and satisfies that $\delta_n \geq C_\epsilon \big(\epsilon_n^2\sqrt{n} \vee n^{-1/4} \vee n^{-(1-2/q)}\big)$. Such construction is possible since $\sqrt{n}\epsilon_n^2=o(1)$. We set the nuisance realization set $N_n$ (denoted by $\mathcal{T}_N$ in the DML paper) to be the set of all vector functions $(Q_t,P_t,\pi_z:z \in \mathcal{Z})$ consisting of square-integrable functions $Q_{t,z},P_{t,z},$ and $\pi_z$ such that for all $z \in \mathcal{Z}$:
	\begin{align*}
		& \norm{Q_{t,z}}_q \leq C, P_{t,z} \in [0,1], \pi_z \in [\epsilon,1], z \in \mathcal{Z}, \\
		&  \lVert Q_{t,z} - Q_{t,z}^o \rVert_q \vee \lVert P_{t,z} - P_{t,z}^o \rVert_q \vee  \norm{\pi_{z} - \pi_{z}^o}_q  \leq \epsilon_n, \\
		& \norm{\pi_{z} - \pi_{z}^o}_2 \times \big( \lVert Q_{t,z} - Q_{t,z}^o \rVert_2 + \lVert P_{t,z} - P_{t,z}^o \rVert_2 \big) \leq \epsilon_n^2.
	\end{align*}

	Consider Assumption 3.1 in the DML paper. Assumption 3.1(d), the Neyman orthogonality condition, is verified by Proposition \ref{prop:Neyman}, where the validity of the differentiation under the integral operation is verified later in the proof. Assumption 3.1(e), the identification condition, is verified by the condition that $p_{t,k}^o \in [\epsilon,1]$. The remaining conditions of Assumption 3.1 in the DML paper are trivially verified. 
	
	Next, we consider Assumption 3.2 in the DML paper. Note that Assumption 3.2(a) holds by the construction of $N_n$ and $\epsilon_n$ and our assumptions on the nuisance estimates. Assumption 3.2(d) is verified by our assumption that the semiparametric efficiency bound of $\beta_{t,k}$ is above $\epsilon$.
	The remaining task is to verify Assumption 3.2(b) and 3.2(c) in the DML paper. To do that, we choose $n$ sufficiently large and let $(Q_{t,z},P_{t,z},\pi_z:z \in \mathcal{Z})$ be an arbitrary element of the nuisance realization set $N_n$. We keep the above notations throughout the remaining part of the proof. Define
	\begin{align*}
		\psi^a_z(T,Z,X,P_t,\pi_z) = \frac{\mathbf{1}\{Z=z\}}{\pi_z(X)}(\mathbf{1}\{T=t\} - P_{t,z}(X)) + P_{t,z}(X)
	\end{align*}
	and
	\begin{align*}
		\psi^b_z(Y,T,Z,X,Q_t,\pi_z) = \frac{\mathbf{1}\{Z=z\}}{\pi_z(X)}(Y\mathbf{1}\{T=t\} - Q_{t,z}(X)) + Q_{t,z}(X).
	\end{align*}
	Since $\psi^a$ is a linear combination of $\psi^a_z,z\in\mathcal{Z}$ and $\psi^b$ is a linear combination of $\psi^b_z,z\in\mathcal{Z}$, we only need $\norm{\psi^a_z(T,Z,X,P_t,\pi_z)}_q$ and $\norm{\psi^b_z(Y,T,Z,X,Q_t,\pi_z)}_q$ to be uniformly bounded (i.e., the bounds do not depend on $n$) for $z \in \mathcal{Z}$ in order to verify Assumption 3.2(b) in the DML paper. In fact, 
	\begin{align*}
		\norm{\psi^b_z(Y,T,Z,X,P_t,\pi_z)}_q & \leq \norm{ \mathbf{1}\{Z=z\} / \pi_z(X) \abs{Y\mathbf{1}\{T=t\} - Q_{t,z}(X)}}_q + \norm{Q_{t,z}(X)}_q \\
		& \leq \frac{1}{\epsilon}\left( \norm{Y\mathbf{1}\{T=t\}}_{q} + \norm{Q_{t,z}(X)}_{q} \right) + \norm{Q_{t,z}(X)}_q \leq 2C/\epsilon + C,
	\end{align*}
	where we have used the assumption that $\pi_z \geq \epsilon$, $\norm{Y\mathbf{1}\{T=t\}}_q \leq C$, and $\norm{Q_t(X)}_q \leq C$. Similarly, we have
	\begin{align*}
		\norm{\psi^a_z(T,Z,X,P_t,\pi_z)}_q & \leq \norm{ \mathbf{1}\{Z=z\} / \pi_z(X) \abs{\mathbf{1}\{T=t\} - P_{t,z}(X)}}_q + \norm{P_{t,z}(X)}_q \\
		& \leq \frac{1}{\epsilon}\big( 1 + \norm{P_{t,z}(X)}_{q} \big) + \norm{P_{t,z}(X)}_q \leq 2/\epsilon + 1,
	\end{align*}
	where we have used the assumption that $\pi_z \geq \epsilon$ and $P_t \in [0,1]$. Thus, Assumption 3.2(b) in the DML paper is verified. 

	To verify Assumption 3.2(c) in the DML paper, we again only need to verify the corresponding conditions for $\psi^a_z$ and $\psi^b_z$, respectively. For $\psi^a_z$, we have
	\begin{align*}
		& \norm{\psi^a_z(T,Z,X,P_t,\pi_z) - \psi^a_z(T,Z,X,P_t^o,\pi_z^o)}_2 \\
		\leq & \norm{\frac{\pi_z(X) - \pi_z^o(X)}{\pi_z(X) \pi_z^o(X)}}_2 + \norm{\frac{P_{t,z}(X)}{\pi_z(X)} - \frac{P_{t,z}^o(X)}{\pi_z^o(X)} }_2 + \norm{P_{t,z}(X) -P_{t,z}^o(X)}_2 \\
		\leq & \frac{1}{\epsilon^2} \norm{\pi_z(X) - \pi_z^o(X)}_2 + \frac{1}{\epsilon^2} \norm{(P_{t,z}(X) - P_{t,z}^o(X))\pi_z^o(X) + P_{t,z}^o(X)(\pi_z^o(X) - \pi_z(X))}_2 \\
		 & + \norm{P_{t,z}(X) -P_{t,z}^o(X)}_2 \\
		\leq & \frac{2}{\epsilon^2} \norm{\pi_z(X) - \pi_z^o(X)}_2 + \left( 1/\epsilon^2 + 1 \right) \norm{P_{t,z}(X) -P_{t,z}^o(X)}_2 \leq C_{\epsilon,1} \varepsilon_n \leq \delta_n,
	\end{align*}
	where the second to last inequality follows from the fact that $P_{t,z}^o,\pi_z^o \in [0,1]$. For $\psi^b_z$, we have 
	\begin{align*}
		& \norm{\psi^b_z(Y,T,Z,X,Q_t,\pi_z) - \psi^b_z(Y,T,Z,X,Q_t^o,\pi_z^o)}_2 \\
		\leq & \frac{1}{\epsilon^2} \norm{\pi_z^o(X)(Y\mathbf{1}\{T=t\} - Q_{t,z}(X)) - \pi_z(X)(Y\mathbf{1}\{T=t\} - Q_{t,z}^o(X))}_2 \\
		& + \norm{Q_{t,z}(X) -Q_{t,z}^o(X)}_2 \\
		= & \frac{1}{\epsilon^2} \norm{(Y\mathbf{1}\{T=t\} - Q_{t,z}^o(X))(\pi^o_z(X) - \pi_z(X)) + \pi_z^o(X)(Q_{t,z}^o(X) - Q_{t,z}(X))}_2 \\
		& + \norm{Q_{t,z}(X) -Q_{t,z}^o(X)}_2 \\
		\leq & \frac{1}{\epsilon^2} \norm{(Y\mathbf{1}\{T=t\} - Q_{t,z}^o(X))(\pi^o_z(X) - \pi_z(X)) }_2 + \norm{\pi_z^o(X)(Q_{t,z}^o(X) - Q_{t,z}(X))}_2 \\
		& + \norm{Q_{t,z}(X) -Q_{t,z}^o(X)}_2 \\
		\leq & \frac{C}{\epsilon^2} \norm{\pi^o_z(X) - \pi_z(X)}_2 + \left( \frac{1}{\epsilon^2} + 1\right) \norm{Q_{t,z}^o(X) - Q_{t,z}(X)}_2 \leq C_{\epsilon,2}\varepsilon_n \leq \delta_n,
	\end{align*}
	where the last inequality follows from our assumption that $|Y\mathbf{1}\{T=t\} - Q_t^o(X)| \leq C$ and the fact that $\pi_z^o \in [\epsilon,1]$. Combining the above two inequality results, we can verify the first two conditions of Assumption 3.2(c) in the DML paper. 
	
	For the last condition of Assumption 3.2(c) in the DML paper, which bounds the second-order Gateaux derivative, we again consider $\psi^a_z$ and $\psi^b_z$ separately. For $r \in [0,1)$, recall that $Q_{t,z}^r = Q_{t,z}^o + r(Q_{t,z} - Q_{t,z}^o),$ $P_{t,z}^r = P_{t,z}^o + r(P_{t,z} - P_{t,z}^o),$ and $\pi^r_z = \pi^o_z + r(\pi_z - \pi^o_z)$. Clearly, $P_{t,z}^r,\pi_z^r \in [0,1]$. With differentiation under the integral, we have
	\begin{align*}
		& \frac{\partial^2}{\partial r^2} \mathbb{E} \left[ \psi^a_z(T,Z,X,P_t^r,\pi_z^r) \right] \\
	= & \frac{\partial}{\partial r} \mathbb{E} \Big[ \frac{ -\mathbf{1}\{Z = z\} }{(\pi^r_z(X))^2} \left( \mathbf{1}\{T=t\} - P^r_{t,z}(X) \right) \left( \pi_z(X) - \pi^o_z(X) \right) \\
	& +  P_{t,z}(X) - P^o_{t,z}(X) - \frac{\mathbf{1}\{Z = z\} }{\pi^r_z(X)} \left( P_{t,z}(X) - P^o_{t,z}(X) \right)  \Big] \\
	= & \mathbb{E} \Big[ \frac{2 \times \mathbf{1}\{Z=z\}}{(\pi_z^r(X))^3}(\pi_z(X) - \pi_z^o(X))^2(\mathbf{1}\{T=t\} - P_{t,z}^r(X)) \Big] \\
	& + \mathbb{E} \Big[ \frac{\mathbf{1}\{Z=z\}}{(\pi_z^r(X))^2}(\pi_z(X) - \pi_z^o(X))(P_{t,z}(X) - P_{t,z}^o) \Big] \\
	& + \mathbb{E} \Big[ \frac{\mathbf{1}\{Z=z\}}{(\pi_z^r(X))^2}(\pi_z(X) - \pi_z^o(X))(\mathbf{1}\{T=t\} - P_{t,z}^r(X))(P_{t,z}(X) - P_{t,z}^o) \Big] \\
	& - \mathbb{E} \Big[ \frac{\mathbf{1}\{Z=z\}}{\pi_z^r(X)}(\mathbf{1}\{T=t\} - P_{t,z}^r(X))(P_{t,z}(X) - P_{t,z}^o)^2 \Big].
	\end{align*}
	Using the fact that $|\mathbf{1}\{T=t\} - P_t^r(X)| \leq 1$ and $\pi_z^r \geq \epsilon$, we can bound the above derivative by
	\begin{align*}
		\Big| \frac{\partial^2}{\partial r^2} \mathbb{E} \left[ \psi^a_z(T,Z,X,P_t^r,\pi_z^r) \right] \Big| &\leq C_\epsilon \big( \norm{\pi_z(X) - \pi_z^o(X)}_2^2 + \big\lVert P_{t,z}(X) - P_{t,z}^o(X) \big\rVert_2^2 \big) \\
		& \quad + C_\epsilon \norm{\pi_z(X) - \pi_z^o(X)}_2 \times \lVert P_{t,z}(X) - P_{t,z}^o(X) \rVert_2 \\
		& \leq C_{\epsilon,3} \varepsilon_n^2 \leq \delta_n / \sqrt{n}.
	\end{align*}
	By bounding the first and second derivative uniformly with respect to $r$, we know that the differentiation under the integral operation is valid. So the Neyman orthogonality condition is verified. Analogously, we can show that
	\begin{align*}
		& \frac{\partial^2}{\partial r^2} \mathbb{E} \left[ \psi^b_z(Y,T,Z,X,Q_t^r,\pi_z^r) \right] \\
		= & \mathbb{E} \Big[ \frac{2 \times \mathbf{1}\{Z=z\}}{(\pi_z^r(X))^3}(\pi_z(X) - \pi_z^o(X))^2(Y\mathbf{1}\{T=t\} - Q_{t,z}^r(X)) \Big] \\
	& + \mathbb{E} \Big[ \frac{\mathbf{1}\{Z=z\}}{(\pi_z^r(X))^2}(\pi_z(X) - \pi_z^o(X))(Q_{t,z}(X) - Q_{t,z}^o) \Big] \\
	& - \mathbb{E} \Big[ \frac{\mathbf{1}\{Z=z\}}{(\pi_z^r(X))^2}(\pi_z(X) - \pi_z^o(X))(Y\mathbf{1}\{T=t\} - Q_{t,z}^r(X))(Q_{t,z}(X) - Q_{t,z}^o) \Big] \\
	& - \mathbb{E} \Big[ \frac{\mathbf{1}\{Z=z\}}{\pi_z^r(X)}(Y\mathbf{1}\{T=t\} - Q_{t,z}^r(X))(Q_{t,z}(X) - Q_{t,z}^o)^2 \Big].
	\end{align*}
	Under the assumption $|Y\mathbf{1}\{T=t\} - Q_{t,z}^o(X)| \leq C$, we have 
	\begin{align*}
		|Y\mathbf{1}\{T=t\} - Q_{t,z}^r(X)| \leq |Y\mathbf{1}\{T=t\} - Q_{t,z}^o(X)| + r|Q_{t,z}(X) - Q_{t,z}^o| \leq C + 1, 
	\end{align*}
	for all $r \in [0,1]$ and $n$ large enough. Then we can bound the above derivative by
	\begin{align*}
		\Big| \frac{\partial^2}{\partial r^2} \mathbb{E} \left[ \psi^b_z(Y,T,Z,X,Q_t^r,\pi_z^r) \right] \Big| \leq & C_\epsilon \big( \norm{\pi_z(X) - \pi_z^o(X)}_2^2 + \big\lVert Q_{t,z}(X) - Q_{t,z}^o(X) \big\rVert_2^2 \big) \\
		& + C_\epsilon \norm{\pi_z(X) - \pi_z^o(X)}_2 \times \norm{Q_{t,z}(X) - Q_{t,z}^o(X)}_2 \\
		\leq & C_{\epsilon,4} \varepsilon_n^2 \leq \delta_n/\sqrt{n}.
	\end{align*}
	Therefore, we have verified the last condition of Assumption 3.2(c) in the DML paper.

	Lastly, we need to verify the condition on $\delta_n$ in Theorem 3.1 and 3.2 in the DML paper, that is, $\delta_n \geq n^{-[(1-2/q)\wedge(1/2)]}$. This directly follows from the construction of $\delta_n$.
\end{proof}

\subsection{Proof of Weak IV Inference Results}

\begin{proof} [Proof of Theorem \ref{thm:weak-id-test}]
	We first prove part (i). Consider applying the DML method to the moment condition (\ref{eqn:WI-moment-condition}) 
		to estimate the parameter $\upsilon - \beta_0 p$ and obtain the standard error. We want to show the convergence in distribution of
		\begin{align} \label{eqn:normalized-robust-statistic}
			\check{\sigma}_{\psi}^{-1} \sqrt{n} \left[ (\check{\upsilon} - \beta_0 \check{p}) - (\upsilon - \beta_0 p) \right] = \check{\rho} - \sqrt{n}(\upsilon - \beta_0 p)/\check{\sigma}_{\psi}
		\end{align}
		to the standard normal distribution uniformly over the DGPs in $\mathcal{P}^{\text{WI}}(c_0,c_1)$. To do that, we need to verify Assumptions 3.1 and 3.2 in the DML paper regarding the above moment condition. Assumptions 3.1(a)-(c) hold trivially. Assumption 3.1(d), the Neyman orthogonality condition, is verified by Proposition \ref{prop:Neyman}. That is, the Gateaux derivatives with respect to the nuisance parameters are zero regardless of the value of $\beta$. Assumption 3.1(e), the identification condition, is verified since the Jacobian of the parameter in the moment condition is $1$. Assumption 3.2 in the DML paper can be verified in the same way as in the proof of Theorem \ref{thm:dml}. For brevity, we do not repeat the verification here.

		For DGPs in $\mathcal{P}^{\text{WI}}_{\beta_0}(c_0,c_1)$, (\ref{eqn:normalized-robust-statistic}) is equal to $\check{\rho}$. Therefore, the uniform convergence in distribution of $|\check{\rho}|$ is established in the null space, and the size of the test is uniformly controlled accordingly. For DGPs in $\mathcal{P}^{\text{WI}}_{\beta}(c_0,c_1)$, where $\beta > \beta_0$, we have
		\begin{align*}
			\check{\rho} & = \left( \check{\rho} - \sqrt{n}(\upsilon - \beta_0 p)/\check{\sigma}_{\psi}  \right) + \sqrt{n}(\upsilon - \beta_0 p)/\check{\sigma}_{\psi} \\
			&= \left( \check{\rho} - \sqrt{n}(\upsilon - \beta_0 p)/\check{\sigma}_{\psi}  \right) + \sqrt{n}(\beta - \beta_0)p/\check{\sigma}_{\psi}.
		\end{align*}
	The first term on the RHS of the last equality converges in distribution to $N(0,1)$. In contrast, the second term diverges to infinity since $\check{\sigma}_{\psi}$ converges in probability to $\sigma_\psi \geq \sqrt{c_0}$ by Theorem 3.2 in the DML paper. Therefore, the probability of $|\check{\rho}|$ exceeding any finite number converges to 1. The case where $\beta < \beta_0$ is essentially the same.

	To prove part (ii) of the theorem, notice that $(\beta - \beta_0)p \leq 0$ for any DGP in the null space $\bigcup_{\beta \leq \beta_0}\mathcal{P}^{\text{WI}}_{\beta}(c_0,c_1)$, which implies that $\check{\rho} \leq \check{\rho} - \sqrt{n}(\upsilon - \beta_0 p)/\check{\sigma}_{\psi}$. Therefore,
	\begin{align*}
		\sup_P \mathbb{P}_P\big(\check{\rho} > \mathcal{N}_{1-\alpha} \big) \leq \sup_P \mathbb{P}_P\left(\check{\rho}- \sqrt{n}(\upsilon - \beta_0 p)/\check{\sigma}_{\psi} > \mathcal{N}_{1-\alpha} \right) \rightarrow \alpha,
	\end{align*}
	where the supremum is taken over $P \in \bigcup_{\beta \leq \beta_0}\mathcal{P}^{\text{WI}}_{\beta}(c_0,c_1)$. Consistency can be derived in the same way as part (i).
\end{proof}

\section{Implicitly Defined Parameters} \label{sec:spe2}

This section studies general parameters defined implicitly through moment conditions. We allow the moment conditions to be non-smooth, which is the case when the parameter of interest is the quantile. We also allow the moment conditions to be overidentifying, which could be the result of imposing the underlying economic theory on multiple levels of treatment and instrument.


To facilitate the exposition, we define a random variable $ Y^*_{t,k}$ such that the marginal distribution of $Y^*_{t,k}$ is equal to the conditional distribution of $Y_t $ given $S \in \Sigma_{t,k}$. The joint distribution of the $Y^*_{t,k}$'s is irrelevant and hence left unspecified. For convenience, we use a single index $j \in J$ rather than $(t,k)$ for labeling. That is, we collect the $Y^*_{t,k}$'s into the vector $Y^* \equiv (Y^*_1,\cdots,Y^*_J)$. Let $t_j$ be the treatment level associated with $Y^*_j$. The quantities $p_j $ and $b_j$ are analogously defined.\footnote{We can further extend the vector $Y^*$ to include variables whose marginal distributions are the same as the conditional distributions of $Y_t$ given $T=t, S \in \Sigma_{t,k}$. Efficient estimation in this more general case is similar and hence omitted for brevity.}

Let the parameter of interest be $\eta$, which lies in the parameter space $\Lambda \subset \mathbb{R}^{d_\eta}$, $d_\eta \leq J$. The true value of the parameter $\eta_0$ satisfies the moment condition
\begin{equation*}
	\mathbb{E} \left[ m(Y^*,\eta^o) \right] = 0,
\end{equation*}
where $m: \mathcal{Y}^{J} \times \mathbb{R}^{d_\eta} \rightarrow \mathbb{R}^{J}$ is a vector of functions:
\begin{align*}
	m(Y^*,\eta) \equiv \left(m_1(Y_1^*,\eta), \cdots, m_{J}(Y_{J}^*,\eta)  \right)'
\end{align*}
Since the vector $\eta$ appears in each $m_j$, restrictions are allowed both within and across different subpopulations. Another interesting feature of this specification is that the moment conditions are defined for the random variables that are not observed. But their marginal distributions can be identified similar to Theorem \ref{thm:id1}. 

Let $\bar{m} \equiv (\bar{m}_{1}', \cdots, \bar{m}_{J}')'$, where 
\begin{align*}
	\bar{m}_{j}(X,\eta) = \left(\bar{m}_{j,z_1}(X,\eta), \cdots,\bar{m}_{j,z_{N_Z}}(X,\eta)  \right)' 
\end{align*}
and 
\begin{align*}
	\bar{m}_{j,z}(X,\eta) = \mathbb{E}\left[ m_j(Y,\eta) \mathbf{1}\{T = t_j\} \mid Z=z, X \right].
\end{align*}
The functions $\bar{m}_{j,z}$ are identified from the data.
Similar to Theorem \ref{thm:id1}, we can show that the parameter $\eta$ is identified by the moment conditions:
\begin{align*}
	b_{j} \mathbb{E} \left[ \bar{m}_{j}(X,\eta) \right] = 0, 1 \leq j \leq J \iff \eta = \eta^o.
\end{align*}
The following theorem gives the SPEB for the estimation of $\eta$. 

\begin{theorem} \label{thm:speb2}
	Assume the following conditions hold.
	\begin{enumerate}[label = (\roman*)]
		\item $\mathbb{E} \left[ m(Y^*,\eta)^2 \right] < \infty, \eta \in \Lambda $.
		\item \label{gmm_con_diff} For each $j$ and $z$,  $m_{j,t_j,z}$ is continuously differentiable in its second argument. Let $ \Gamma$ be the $J \times d_\eta$ matrix whose $j$th row is $ b_{j} \frac{d}{d \eta} \mathbb{E} \left[ \bar{m}_{j}(X,\eta) \right] \big|_{\eta = \eta^o}'  $, and assume $\Gamma$ has full column rank.
	\end{enumerate}
	Then for the estimation of $\eta$, the EIF is	
	\begin{equation} \label{gmmeif}
		- \left( \Gamma' V^{-1} \Gamma \right)^{-1} \Gamma' V^{-1}  \psi^\eta(Y,T,Z,X,\eta^o,\pi^o,\bar{m}^o),
	\end{equation}
	where 
	\begin{align*}
		V = \mathbb{E}\left[  \psi^\eta(Y,T,Z,X,\eta,\pi,\bar{m}) \psi^\eta(Y,T,Z,X,\eta,\pi,\bar{m})' \right] 
	\end{align*}
	and $\psi^\eta(Y,T,Z,X,\eta,\pi,\bar{m})$ is a $J \times 1$ random vector whose $j$th element is
	\begin{equation} \label{Psi_m}
		b_{j} \left( \zeta(Z,X,\pi)  \left( \iota (m_j(Y,\eta)\mathbf{1}\{T=t_j\}) - \bar{m}_{j}(X,\eta) \right) + \bar{m}_{j}(X,\eta) \right)
	\end{equation}
	In particular, the semiparametric efficiency bound is $\left( \Gamma' V^{-1} \Gamma \right)^{-1}$.
\end{theorem}

\begin{proof} [Proof of Theorem \ref{thm:speb2}]
	The proof is based on the approach described in section 3.6 of \cite{hong2010semiparametric} and the proof of Theorem 1 in \cite{cattaneo2010efficient}. We use a constant $d_\eta \times d_m$ matrix $A$ to transform the overidentified vector of moments into an exactly identified system of equations $A \left( b_{j} \mathbb{E} \left[ \bar{m}_{j}(X,\eta) \right] \right)_{j=1}^J  = 0$, find the $A$-dependent EIF for the exactly-identified parameter, and choose the optimal $A$. In a parametric submodel, the implicit function theorem gives that
	\begin{align*}
		\frac{ \partial }{ \partial \theta } \eta \big|_{\theta = \theta^o} = - \left(A \Gamma \right)^{-1} A \frac{ \partial }{ \partial \theta } \left( b_{j} \mathbb{E}_\theta \left[ \bar{m}_{j}(X,\eta^o) \right] \right)_{j=1}^J\big|_{\theta = \theta^o},
	\end{align*}
	where $\frac{ \partial }{ \partial \theta } \mathbb{E}_\theta \left[ \bar{m}_{j}(X,\eta^o) \right] \big|_{\theta = \theta^o}$ is an $N_Z \times 1$ random vector whose typical element can be represented by
	\begin{align*}
		& \int m_j(y,\eta^o) \mathbf{1}\{\tau = t_j\} s_{z}(y,\tau \mid x; \theta^o) f_{z}(y,\tau \mid x ; \theta^o) f_X(x ; \theta^o) dyd\tau dx \\
		+ &  \int m_j(y,\eta^o) \mathbf{1}\{\tau = t_j\} s_X( x; \theta^o) f_{z}(y,\tau \mid x ; \theta^o) f_X(x ; \theta^o) dyd\tau dx,
	\end{align*}
	for $z \in \mathcal{Z}$.
	So the EIF for this exactly-identified parameter is 
	\begin{align*}
		\psi^A(Y,T,Z,X,\eta^o,\pi^o,\bar{m}^o) = - \left(A \Gamma \right)^{-1} A \Psi^\eta(Y,T,Z,X,\eta^o,\pi^o,\bar{m}^o),
	\end{align*}
	where $\psi^\eta$ is defined by Equation (\ref{Psi_m}). It is straightforward to verify that $\psi^A$ satisfies $\frac{\partial }{ \partial \theta} \eta \big|_{\theta = \theta^o} = \mathbb{E} \left[ \psi^A s_{\theta^o}' \right] \text{, and } \psi^A \in \mathscr{S}$. The optimal $A$ is chosen by minimizing the sandwich matrix $\mathbb{E} \left[ \psi^A (\psi^A)' \right] =  \left(A \Gamma\right)^{-1} A \mathbb{E} \left[ \psi^\eta (\psi^\eta)' \right] A' \left(\Gamma' A' \right)^{-1}  $.
	Thus, the EIF for the over-identified parameter is obtained when $A = \Gamma' V^{-1}$. Plugging this expression into $\psi^A$, we obtain Equation (\ref{gmmeif}).
\end{proof}

Note that, for example, $m_j(Y^*_{j},\eta) = Y^*_{j} - \eta$, then $\eta=\beta_{j}$, and the efficiency bound shown above reduces to the one computed in Theorem \ref{thm:speb1}. If $T=Z$, that is, the treatment satisfies the unconfounded, then the Theorem \ref{thm:speb2} reduces to Theorem 1 in \cite{cattaneo2010efficient}. 

For estimation, we use the EIFs to generate moment conditions and propose a three-step semiparametric GMM procedure. The criterion function is
\begin{equation}
	\Psi^{\eta}_{n}(\eta,\pi,m) = \frac{1}{n}\sum_{i=1}^n \psi^\eta(Y_i,T_i,Z_i,X_i, \eta,\pi,\bar{m}).
\end{equation}
Its probability limit is denoted as
\begin{equation}
	\Psi^\eta(\eta,\pi,m_{Z}) = \mathbb{E} \left[ \psi^\eta(Y,T,Z,X, \eta,\pi,\bar{m}) \right],
\end{equation}
where the expectation is taken with respect to the true parameters $(\pi^o,\bar{m}^o)$. The implementation procedure is as follows. Assume that we have nonparametric estimators $\hat{\pi}$ and $\hat{m}$ that consistently estimate $\pi^o$ and $\bar{m}^o$, respectively. We first find a consistent GMM estimator $\tilde{\eta}$ using the identity matrix as the weighting matrix, that is,
\begin{equation}
	\norm{\Psi^\eta_n(\tilde{\eta},\hat{\pi},\hat{m})}_2 \leq \inf_{\eta \in \Lambda} \norm{\Psi^\eta_n(\eta,\hat{\pi},\hat{m})}_2 + o_p(1).
\end{equation}
Next, we use this estimate to form a consistent estimator $\hat{V}$ of the covariance matrix $V$, where 
\begin{align*}
	\hat{V} = \frac{1}{n} \sum_{i=1}^n \psi^\eta(Y_i,T_i,Z_i,X_i, \tilde{\eta},\hat{\pi},\hat{m}) \psi^\eta(Y_i,T_i,Z_i,X_i, \tilde{\eta},\hat{\pi},\hat{m})'.
\end{align*}
Then we let $\hat{\eta}$ be the optimally-weighted GMM estimator:
\begin{align*}
	& \Psi^\eta_n(\hat{\eta},\hat{\pi},\hat{m}_Z) V_n(\tilde{\eta},\hat{\pi},\hat{m}_Z)^{-1} \Psi^\eta_n(\hat{\eta},\hat{\pi},\hat{m}_Z)' \\
	\leq & \inf_{\eta \in \Lambda} \Psi^\eta_n(\eta,\hat{\pi},\hat{m}_Z) V_n(\tilde{\eta},\hat{\pi},\hat{m}_Z)^{-1} \Psi^\eta_n(\eta,\hat{\pi},\hat{m}_Z)' + o_p \big( n^{-1/2} \big).
\end{align*}
To conduct inference, we estimate $\Gamma$ using the estimator $\hat{\Gamma}$ whose elements are defined as
\begin{align*}
	\hat{\Gamma}_{jl} =  \frac{1}{n} \sum_{i=1}^n b_{j} \frac{\partial}{\partial \eta}\hat{m}_{j}(X_i,\eta) \Big|_{\eta = \hat{\eta}},
\end{align*}
where we have implicitly assumed that the estimator $\hat{m}_{j}$ is differentiable in its second argument.

In the following theorem, we derive the asymptotic properties of the GMM estimators. The main theoretical difficulty is that the random criterion function $\Psi_n(\cdot,\hat{\pi},\hat{m})$ could potentially be discontinuous because we allow $m(Y^*,\cdot)$ to be discontinuous. We use the theory developed in \cite{chen2003estimation} to overcome this problem.\footnote{\cite{cattaneo2010efficient} instead uses the theory from \cite{pakes1989simulation}. However, the general theory of \cite{chen2003estimation} is more straightforward to apply in this case since they explicitly assume the presence of infinite-dimensional nuisance parameters, which can depend on the parameters to be estimated.} Let $\Pi_z$ be the function class that contains $\pi_z^o$. Let $\mathcal{M}_{j,z}$ be the function class that contains $\bar{m}_{j,z}^o$.

\begin{theorem} \label{thm:est2}
	Let the assumptions in Theorem \ref{thm:speb2} hold. Assume the following conditions hold.
	\begin{enumerate}[ label = (\roman*)]
		\item The parameter space $\Lambda$ is compact. The true parameter $\eta^o$ is in the interior of $\Lambda$.
		\item For any $j,z$ and $\bar{m}_{j,z} \in \mathcal{M}_{j,z}$, there exists $C>0$ such that for $\delta>0$ sufficiently small,
		\begin{align*}
			\sup_{\abs{\eta' -\eta} \leq \delta}\mathbb{E}\abs{\bar{m}_{j,z}(X,\eta') - \bar{m}_{j,z}(X,\eta)}^2 \leq C \delta^2.
		\end{align*}
		\item Donsker properties:
		\begin{align*}
			\int_0^\infty \log N(\varepsilon,\Pi_z,\norm{\cdot}_\infty) d\varepsilon, \int_0^\infty \log N(\varepsilon,\mathcal{M}_{j,z},\norm{\cdot}_\infty) d\varepsilon < \infty,
		\end{align*}
		where $N(\varepsilon,\mathcal{F},\norm{\cdot})$ denotes the covering number of the space $(\mathcal{F},\norm{\cdot})$.
		\item Convergence rates of the nonparametric estimators: 
		\begin{align*}
			\norm{\hat{\pi}_z - \pi^o_z}_\infty ,\lVert \hat{m}_{j,z} - \bar{m}_{j,z}^o \rVert_\infty = o_p(n^{-1/4}).
		\end{align*}
		\item The function $\sup_{\eta \in \Lambda} \abs{ \frac{\partial}{\partial \eta} \bar{m}^o_{j}(\cdot,\eta) }$ is integrable. The estimator $\frac{\partial}{\partial \eta}\hat{m}_{j}$ is consistent uniformly in its second argument, that is,
		\begin{align*}
			\norm{\frac{\partial}{\partial \eta}\hat{m}_{j}(x,\eta) - \frac{\partial}{\partial \eta}\bar{m}^o_{j}(x,\eta)}_\infty = o_p(1), \forall x.
		\end{align*}
	\end{enumerate}
	Then $\tilde{\eta}=\eta^o + o_p(1)$, $\hat{V} = V + o_p(1)$, $\hat{\Gamma} = \Gamma + o_p(1)$, and
	\begin{align*}
		\sqrt{n}\left( \hat{\eta} - \eta^o \right) \implies N\left( \bm{0}, (\Gamma'V^{-1}\Gamma)^{-1} \right),
	\end{align*}
	where $\bm{0}$ denotes a vector of zeros.
\end{theorem}

The following lemma is helpful for proving Theorem \ref{thm:est2}.

\begin{lemma} \label{lm:se-clk}
	Under the assumptions of Theorem \ref{thm:speb2}, the class 
	\begin{align*}
		\mathcal{F} \equiv \left\{ \psi^\eta(Y,T,Z,X,\eta,\pi,\bar{m}): \pi \in \Pi_z, \bar{m}_{j,z} \in \mathcal{M}_{j,z}, 1 \leq j \leq J, z \in \mathcal{Z} \right\}
	\end{align*}
	is Donsker with a finite integrable envelope. The following stochastic equicontinuity condition hold: for any positive sequence $\delta_n = o(1)$,
	\begin{align*}
		\sup \big\{ &\Psi^\eta_n(\eta,\pi,\bar{m}) - \Psi^\eta(\eta,\pi,\bar{m}) - \Psi^\eta_n(\eta^o,\pi^o,m_Z^o): \\
		&\norm{\eta - \eta^o}_2 \vee \norm{\pi - \pi^o}_\infty \vee \norm{\bar{m} - \bar{m}^o}_\infty \leq \delta_n \big\} = o_p \big( n^{-1/2} \big),
	\end{align*}
	where the supremum is taken over $\eta \in \Lambda$, $\pi_z \in \Pi_z$, and $\bar{m}_{j,z} \in \mathcal{M}_{j,z}$.
\end{lemma}

\begin{proof} [Proof of Lemma \ref{lm:se-clk}]
	We first verify that the moment condition $\psi^\eta$ satisfies Condition (3.2) of Theorem 3 in \cite{chen2003estimation} (hereafter CLK). In fact, when $\lVert \bar{m}'_{j,z} - \bar{m}_{j,z} \rVert_\infty \vee \norm{\eta' - \eta}_\infty \leq \delta$, the triangle inequality gives that
	\begin{align*}
		& \mathbb{E}\abs{\bar{m}'_{j,z}(X,\eta') - \bar{m}_{j,z}(X,\eta)}^2 \\
		\leq & 2\mathbb{E}\abs{\bar{m}'_{j,z}(X,\eta') - \bar{m}'_{j,z}(X,\eta)}^2 + 2\mathbb{E}\abs{\bar{m}'_{j,z}(X,\eta) - \bar{m}_{j,z}(X,\eta)}^2\\
		\leq & const \times \delta^2,
	\end{align*}
	where we use the notation $\textit{const}$ to denote a generic constant that may have different values at each appearance. The last inequality follows from the assumption (ii).
	Similarly, we can verify that the remaining terms in $\psi^\eta$ also satisfy the same condition. Therefore, $\psi^\eta$ is locally uniformly $L_2$-continuous, that is,
	\begin{align*}
		\mathbb{E} \big[ \sup \big\{ & \abs{\psi^\eta(Y,T,Z,X,\eta',\pi',\bar{m}') - \psi^\eta(Y,T,Z,X,\eta,\pi,\bar{m})}: \\
		&\norm{\eta' - \eta} \vee \norm{\pi' - \pi}_\infty \vee \norm{\bar{m}' - \bar{m}}_\infty \leq \delta \big\} \big] \leq const. \times \delta^2.
	\end{align*}
	Following the same steps as in the proof of Theorem 3 in CLK (p. 1607), we can show that the bracketing number of $\mathcal{F}$ is bounded by
	\begin{align*}
		&N_{[]}\big(\varepsilon,\mathcal{F},\norm{\cdot}_{L_2}\big) \\
		\leq& N(\varepsilon/const,\Lambda,\norm{\cdot}) \times \prod_{z} N(\varepsilon/const,\Pi_z,\norm{\cdot}) \times \prod_{j,z} N(\varepsilon/const,\mathcal{M}_{j,z},\norm{\cdot}).
	\end{align*}
	Therefore, the bracketing entropy of class $\mathcal{F}$ is bounded by
	\begin{align*}
		&\log N_{[]}\big(\varepsilon,\mathcal{F},\norm{\cdot}_{L_2}\big) \\
		\leq& const \times \Big( \log N(\varepsilon/const,\Lambda,\norm{\cdot}) \vee \max_z \log N(\varepsilon/const,\Pi_z,\norm{\cdot}) \\
		\vee & \max_{j,z} \log N(\varepsilon/const,\mathcal{M}_{j,z},\norm{\cdot})  \Big) .
	\end{align*}
	Under the assumption that $\Lambda$ is compact and 
	\begin{align*}
		\int_0^\infty \log N(\varepsilon,\Pi_z,\norm{\cdot}) d\varepsilon, \int_0^\infty \log N(\varepsilon,\mathcal{M}_{j,z},\norm{\cdot}) d\varepsilon < \infty, \forall j,z,
	\end{align*}
	we have that
	\begin{align*}
		\int_0^\infty \log N_{[]}\big(\varepsilon,\mathcal{F},\norm{\cdot}_{L_2}\big) d\varepsilon < \infty.
	\end{align*}
	This implies that $\mathcal{F}$ is Donsker with a finite integrable envelope. Lastly, as stated in Lemma 1 of CLK, the asserted stochastic equicontinuity condition is implied by the fact that $\mathcal{F}$ is Donsker and $\psi^\eta$ is $L_2$-continuous.
\end{proof}

\begin{proof} [Proof of Theorem \ref{thm:est2}]
	We follow the large sample theory in CLK and set $\theta = \eta$, $h = (\pi,\bar{m})$, $M(\theta,h) = \Psi^\eta(\eta,\pi,\bar{m})$, and $M_n(\theta,h) = \Psi^\eta_n(\eta,\pi,\bar{m})$.

	We first use Theorem 1 in CLK to show the consistency of $\tilde{\eta}$. Condition (1.2) in CLK is satisfied because $\Lambda$ is compact, and $\Psi^\eta(\eta,\pi^o,\bar{m}^o)$ has a unique zero and is continuous by our second condition in Theorem \ref{thm:speb2}. As for Condition (1.3) of CLK, we can easily see from the expression of $\Psi$ that it is continuous with respect to $\bar{m}_{j,z}$ and $\pi_z$ (since $\pi_z$ is bounded away from zero), and the uniformity in $\eta$ follows from the fact that $\mathbb{E} \left[ m(Y^*,\eta) \right] $ is bounded as a function of $\eta$. Condition (1.4) of CLK is satisfied by the assumption of Theorem \ref{thm:est2}. The uniform stochastic equicontinuity condition (1.5) of CLK is implied by Lemma \ref{lm:se-clk}. Therefore, $\tilde{\eta} = \eta^o + o_p(1)$.

	We use Corollary 1 (which is based on Theorem 2) in CLK to show the consistency of $\hat{V}$ and the asymptotic normality of $\hat{\eta}$. Condition (2.2) in CLK is verified by the assumptions of Theorem \ref{thm:speb2}. Similar to the proof of Proposition \ref{prop:Neyman}, we can show that the moment condition $\Psi^\eta$, based on the EIF, satisfies the Neyman orthogonality condition for the nuisance parameters $\pi$ and $m_Z$. In fact, for any $j$ and $z$, we let $\pi_z^r = \pi^o_z(X) + r(\pi_z(X) - \pi^o_z(X))$ and $\bar{m}^r_{j,z}(X,\eta) = \bar{m}^o_{j,z}(X,\eta) + r \big( \bar{m}_{j,z}(X,\eta) - \bar{m}^o_{j,z}(X,\eta) \big)$. Then we have
	\begin{align*}
		&\frac{d}{d r} \mathbb{E} \left[ \frac{\mathbf{1}\{Z = z\}}{\pi^r_z(X)} \left( m_j(Y,\eta) \mathbf{1}\{T=t_j\} -  \bar{m}^r_{j,z}(X,\eta)  \right) + \bar{m}^r_{j,z}(X,\eta) \right] \Bigg|_{r=0} \\
		= \mathbb{E} \Bigg[ & -\frac{\mathbf{1}\{Z = z\}}{\left(  \pi^o_z(X) \right)^2} \left( \pi_z(X) - \pi_z^o(X) \right) \left( m_j(Y,\eta) \mathbf{1} \{T = t_j\} - \bar{m}^o_{j,z}(X,\eta) \right)  \\
		+ & \left( \bar{m}^o_{j,z}(X,\eta)  - \bar{m}_{j,z}(X,\eta) \right)   \left(  \frac{\mathbf{1}\{Z = z\}}{\pi^o_z(X)} - 1 \right) \Bigg] = 0,
	\end{align*}
	where we have applied the law of iterated expectations and used the fact that 
	$$\mathbb{E} \left[ \frac{\mathbf{1}\{Z = z\}}{\pi^o_z(X)} \left( m_j(Y,\eta) \mathbf{1} \{T = t_j\} - \bar{m}^o_{j,z}(X,\eta) \right) \Big| X \right] = 0. $$
	Thus, the path-wise derivative of $\Psi^\eta$ with respect to $h = (\pi,\bar{m})$ is zero in any direction. Hence, Condition (2.3) of CLK is verified. Condition (2.4) in CLK directly follows from our assumptions of Theorem \ref{thm:est2}. The stochastic equicontinuity condition (condition (2.6) in CLK) follows from Lemma \ref{lm:se-clk}. Lastly, condition (2.6) in CLK is verified using the central limit theorem since the path-wise derivative is zero. Due to the presence of $\hat{V}$, we also need the uniform convergence condition in Corollary 1 of CLK, which can be verified by using Lemma \ref{lm:se-clk} and an application of Theorem 2.10.14 of \cite{van1996weak}. 
	
	Lastly, to show the consistency of $\hat{\Gamma}$, we only need to show that 
	\begin{align*}
		\frac{1}{n} \sum_{i=1}^n \frac{\partial}{\partial \eta} \hat{m}_{j,t_j,z}(X_i,\hat{\eta}) \overset{p}{\rightarrow} \mathbb{E} \left[ \frac{\partial}{\partial \eta} \hat{m}_{j,z}(X,\eta^o) \right] = \frac{\partial}{\partial \eta} \mathbb{E} \left[  \hat{m}_{j,z}(X,\eta^o) \right],
	\end{align*}
	where the inequality follows from the differentiation under integral operation which holds under the last assumption of the theorem. The convergence in probability follows from the uniform convergence of $\frac{\partial}{\partial \eta} \hat{m}_{j,z}$ and the consistency of $\hat{\eta}$.
	Therefore, the desired convergence results follow.
\end{proof}

\bibliographystyle{chicago}

\bibliography{references.bib}
\end{document}